\RequirePackage{etex} 
\documentclass[preprint,11pt,authoryear]{elsarticle}

\usepackage{amssymb}
\usepackage{amsmath}

\usepackage[dvipsnames]{xcolor}
\usepackage{mathtools}
\usepackage{amsthm}
\usepackage{graphicx}
\usepackage{subcaption}
\usepackage{enumerate}
\usepackage{mathbbol}
\usepackage{float}
\usepackage{pgfplots}
\usepackage{bbm}
\pgfplotsset{compat=1.18}

\usepackage{booktabs}     
\usepackage{array}         
\usepackage{float}         
\usepackage{adjustbox}     
\renewcommand{\arraystretch}{1.2}

\usepackage{xfrac}
\usepackage{cases}
\usepackage[colorlinks=true, allcolors=blue]{hyperref}
\usepackage{bm}
\usepackage{tikz-network}

\newcommand{\dsum}{\displaystyle\sum}
\usepackage{pgfplots}
\usepackage{caption}
\definecolor{pisa}{RGB}{12,102,124}
\definecolor{indigo}{RGB}{75,0,130}

\theoremstyle{plain}
\newtheorem{theorem}{Theorem}[]
\newtheorem{definition}[theorem]{Definition}
\newtheorem{lemma}[theorem]{Lemma}
\newtheorem{proposition}[theorem]{Proposition}

\newtheorem{example}[theorem]{Example}
\newtheorem{axiom}{Axiom}

\journal{}

\begin{document}

\begin{frontmatter}

\title{Power in Sharing Networks with \textit{a priori} Unions}

\author[Mich,Fra,An]
{Michele Aleandri, Francesco Ciardiello, Andrea Di Liddo}

\affiliation[Mich]{organization={Department of AI, Data and Decision Sciences, Luiss Guido Carli},
            addressline={Viale Pola 12}, 
            city={Roma},
            postcode={00198}, 
            state={Lazio},
            country={Italy}}
\affiliation[Fra]{organization={Department of Economics and Statistics, University of Salerno},
            addressline={Via Papa Paolo Giovanni II}, 
            city={Fisciano},
            postcode={84084}, 
            state={Campania},
            country={Italy}}
\affiliation[An]{organization={Department of Economics, University of Foggia},
            addressline={Via Romolo Caggese 1}, 
            city={Foggia},
            postcode={71122}, 
            state={Puglia},
            country={Italy}}
\begin{abstract}
We introduce and analyze a novel family of power indices tailored for sharing networks in technological markets, where firms operate competitively within, but not across, distinct industrial sectors. In these settings, inter-firm collaboration structures emerge from formal technology licensing agreements. The proposed indices are defined over graphs with a priori unions and combine two key centrality measures—degree-based and rescaled eigenvector centrality—modulated by positive market coefficients that reflect sectoral dynamics.

We first explore the monotonicity properties of these indices, highlighting their responsiveness to local changes in network structure. Interestingly, major economic actors exhibit structural stability when inter-sectoral technological spillovers are minimal. Building on these findings, we provide theoretical underpinnings by characterizing the indices as the Shapley values of a family of coherent and economically interpretable transferable utility (TU) games defined over such graphs. However, for a broad class of network structures, the core of these TU games is often empty, signaling inherent instability in technological sharing arrangements.
Finally, we offer an axiomatic foundation for this family of indices, proving independence of the proposed axioms. This axiomatization extends naturally to exchange networks, even when stage-propagation coefficients are not positive.
\end{abstract}



\end{frontmatter}



\section{Introduction}

Technological sharing networks are complex and evolving structures that facilitate the exchange of knowledge, capabilities, and technological assets among firms, often through licensing agreements \citep{Shapiro1985patent, fagerberg2005oxford}. These networks typically emerge in digitally enabled environments, where firms can collaborate across sectoral boundaries without engaging in direct market competition \citep{arora2006patent}. In the framework we consider, firms operate in technologically distinct markets and produce non-substitutable goods. This structural segmentation implies that firms across different markets are not strategic rivals, thus allowing innovation and knowledge to circulate across sectors without triggering traditional antitrust concerns.

By contrast, within each individual market, firms often compete intensely, making them generally reluctant to share proprietary technologies directly \citep{fosfuri2006licensing}. In such environments, the sustainability of knowledge-sharing practices depends crucially on the strength of Intellectual Property Rights (IPR) regimes and the presence of enforceable technological patents. Weak IPR protection can significantly erode incentives to engage in technology transfer, ultimately discouraging innovation diffusion and undermining the efficiency of the network \citep{arora2006patent}. Moreover, licensing agreements carry an inherent trade-off: while licensees may cannibalize the licensor’s sales, such losses may be partially offset by royalty revenues \citep{fosfuri2006licensing}. Nonetheless, patent royalties are typically too low to serve as a strong standalone incentive for technology transfer—especially in mature technological sectors where they tend to be negligible. Their impact may be more relevant in emerging markets, but remains generally limited as a driver of collaborative behavior.
Given these limitations, our analysis focuses on sublicensing as a strategic mechanism through which firms can benefit from knowledge circulation without exposing their core technologies to direct competitors. In particular, firms may redistribute technologies acquired from external markets (when permitted by the original licensors) to other firms, including rivals within the same market. This creates a form of indirect collaboration that, while occurring in a non-cooperative setting, can contribute to expanding the overall market. In this sense, sublicensing acts as a tacit collusive device, mitigating the challenges of direct competition and supporting technology diffusion under constrained incentives. Such dynamics exemplify the concept of \textit{coopetition}—a hybrid strategic condition in which firms simultaneously compete and collaborate \citep{brandenburger2011co}. Increasingly, this logic characterizes contemporary innovation ecosystems, where enabling technologies—such as cloud infrastructures, blockchain systems, and digital licensing platforms—support more fluid, decentralized, and scalable exchanges of innovation.

Within this broader context, technological networks raise critical questions about value appropriation, intellectual property governance, and the distribution of power—particularly in the presence of high market concentration. A particularly illustrative example is the 2019 Google–Huawei case: following U.S. sanctions, Google revoked Huawei’s Android license, thereby cutting off access to key services such as Google Play and Maps. This disruption profoundly affected Huawei’s global business trajectory, demonstrating how control over essential technological channels within a network can reshape power dynamics and dependencies among actors \citep{gilsing2008network}.
Motivated by these observations, our goal is to develop a formal power index—a metric that captures the degree of influence or control that each node (i.e., firm) exerts within a technological network under homogenous market assumptions. In our framework, connections represent the channels through which technology transfers occur, and the proposed index allows us to quantify the strategic role of each firm in governing the flow of innovation. This provides a systematic approach to analyzing power asymmetries in complex technological environments.

In Section~\ref{sect1}, we introduce a new family of power indices for firms embedded in technological sharing networks, modeled as undirected graphs with \emph{a priori} unions that represent sectoral partitions. We study key structural properties of these indices, with particular attention to monotonicity: the intuition that a firm’s power should be non-decreasing when new edges are added to the network. Interestingly, we show that this property does not hold in general, highlighting subtle power dependencies in local network structures and multiple sharing. We further conceptualize how power indices naturally induce rankings of firms, and show that such rankings—especially at the top—are useful for identifying the degree of technological power concentration. In particular, we demonstrate that, while these rankings typically depend on market average coefficients, they can become invariant in network topologies characterized by low externalities—a property with potential applications in assessing the degree of effective competition. 

In Section~\ref{coope}, we assume that firms cooperate in sharing. Accordingly, we formulate a family of transferable utility games (TU-games) defined on undirected graphs with a priori unions, representing technological cooperation between firms while preserving the lack of competition across technological sectors. TU-games with \emph{a priori} unions serve as a fundamental theoretical framework for networks where coalitions are pre-assigned or tend to be formed. For instance, the Owen value~\citep{owen1977values} extends the Shapley value to environments with \emph{a priori} unions, while the Banzhaf-Owen value~\citep{owen1981modification} analogously generalizes the Banzhaf index. Indeed, this literature on the Shapley value for graphs and lattice structures has been formalised~\citep{grabisch2006capacities, bilbao2012cooperative, grabisch1997k} and recently studied~\citep{khmelnitskaya2016shapley, hellman2018values, li2019myerson, li2020efficient, li2020myerson}. For instance, cooperative and probabilistic values of the traditional graph centrality measures have been studied in~\citep{szczepanski2014centrality} when graphs have a priori unions.
We show that the Shapley value of these games naturally aligns with our power indices. We find the Shapley value using the potential of the Shapley value as a theoretical tool. At the same time, we show that the same result can be reached using the Harsanyi dividends technique, although this method becomes more problematic as the number of edges in the network with a priori unions increases. We illustrate that this family of games has an empty core in most cases due to the proper sub-additivity of the characteristic function, then the Shapley value  provides fair and stable payoffs under cooperation of nodes. Recently, in \citep{van2022degree,van2024degree}, the authors connect social network theory, particularly network centrality measures, with the economic literature on von Neumann–Morgenstern expected utility functions. Their goal is to provide a utility foundation for centrality measures, enabling the comparison of different positions across various networks

In Section~\ref{Axioms}, we shift to a normative perspective, moving beyond deterministic models of technological sharing (also beyond the cooperative model). We introduce a novel axiomatic characterization of the family of power indices, providing a more general theoretical foundation. This approach dates back to~\citep{myerson1977graphs,owen1977values, hart1983endogenous} and, more recently, is developed in~\citep{vazquez1996owen, alonso2009values,van2015values}.  Our axioms, partially inspired by \citep{van2008characterizations} and \citep{karos2015indirect}, can be seen as normative principles, while preserving the functional form of the indices derived in deterministic contexts. Among these, the  Inter-Union Neighborhood axiom plays a central role: it states that a firm's power depends, in a probabilistic sense, only on its immediate links to firms in different technological domains. This captures the idea that influence is rooted in the firm's capacity to propagate technologies across sectoral boundaries, consistently with the licensing frameworks discussed earlier.
Importantly, our axiomatic result also holds in more general settings where the coefficients of our indices may take a different interpretation. While such cases do not arise in our baseline framework, where coefficients represent average market values and are reasonably non-negative, this broader formulation is conceptually meaningful in domains where the propagation of sharing of ``items'' can produce adverse effects.
In fact, there are important cases where the impact of multiple sharing cannot be considered always positive or negative: it may be initially beneficial, but become harmful beyond a certain stage of propagation. One example is the diffusion of a new communication standard or API (Application Programming Interface): at first, widespread adoption increases interoperability, but excessive dominance can lead to technological lock-in or systemic fragility. Similarly, in epidemic dynamics, controlled exposure to a pathogen (e.g., in the form of targeted immunity strategies) can help build population resilience or prevent more dangerous variants from dominating. Yet, if the contagion (or the sharing in our former language) exceeds containment capacity, the same epidemic network that initially enabled preparedness can fuel a large-scale health crisis.
In such contexts, modeling of propagation effects—including the use of negative or opposite coefficients in sign—is essential. Our axiomatic framework accommodates these possibilities. Nevertheless, we finally provide a coherent and final result for positive power indices in the analyzed context of technological sharing networks. This is achieved through extremely mild assumptions on our axioms, ensuring that the indices remain interpretable under our settings from the normative point of view.

All proofs are provided in the Appendix \ref{app:proofs1}.

\section{A family of power indices in Sharing Networks}\label{sect1}
In our framework, firms are embedded in \emph{a priori} sectors, each representing a distinct and mutually exclusive technological market. Every firm is assumed to be the exclusive owner of a proprietary technology.
Technology sharing occurs through licensing agreements or formal transfers. A firm—referred to as the \emph{licensor}—may grant a license for its technology to another firm—the \emph{licensee}—typically operating in a different technological domain. This license authorises the licensee to use, adapt, or integrate the technology, subject to the terms and limitations of the agreement. The licensed asset may include a patented process, a software tool, or a specific manufacturing technique.
With the licensor’s consent, the licensee may further grant a sublicense to a third party. A sublicense allows the sublicensee to access and use the technology under the original license’s constraints. Sublicenses can be issued either within the same technological sector or across different ones, thereby promoting the diffusion of innovation beyond sectoral boundaries.
Crucially, sublicenses do not generally confer commercialisation rights; rather, they permit restricted use under clearly defined conditions. These limitations serve to safeguard the licensor’s competitive advantage and intellectual property while allowing for a controlled and strategic dissemination of technological capabilities across firms and industries.
Technological interactions among firms are represented by undirected graphs, where nodes denote firms and edges represent mutual access to technologies. 

We formalise these structures using undirected graphs with coalition structures, as described below.
Let \( N \) be a finite set of \( n \) firms (or nodes), where \( n \in \mathbb{N} \) and \( n \geq 3 \), and let \( E \) denote the set of undirected edges connecting elements of \( N \). Formally,
\[
E \subseteq \mathbf{E}(N) = \big\{ \{i, j\} : i, j \in N \big\}.
\]
Two nodes \( i \) and \( j \) are said to be neighbours if \( \{i, j\} \in E \).
Firms are organised into predefined groups, each representing a distinct \emph{technological sector}. Let \( r \in \mathbb{N} \), with \( 2 \leq r \leq n \), denote the number of mutually exclusive sectors. The set of such groupings is denoted by \( \mathbf{\Pi}_r(N) \), the collection of all \( r \)-partitions of \( N \), formally defined as
\[
\mathbf{\Pi}_r(N) = \left\{ \Pi = (S_1, \ldots, S_r) : \bigcup_{\ell=1}^r S_\ell = N,\ S_\ell \neq \emptyset,\ S_\ell \cap S_m = \emptyset\ \text{for all } \ell \neq m \right\}.
\]
Each subset \( S_\ell \) in a partition \( \Pi = (S_1, \ldots, S_r) \in \mathbf{\Pi}_r(N) \) is called an a priori union, or simply a union.
Given a set of links \( E \subseteq \mathbf{E}(N) \) and a partition \( \Pi \in \mathbf{\Pi}_r(N) \), we define the pair
\[
G_r = (E, \Pi)
\]
as a graph with a priori unions. The collection of all such graphs is denoted by
\[
\mathcal{G}(N) = \left\{ G_r = (E, \Pi) : E \subseteq \mathbf{E}(N),\ \Pi \in \mathbf{\Pi}_r(N),\ r \leq n \right\}.
\]
This framework captures, for instance, scenarios in which firms from distinct sectors—such as automotive, semiconductors, and telecommunications—own proprietary technologies and selectively engage in cross-sector technology sharing.

\begin{example}\label{ex:threeU}
Car manufacturers such as Toyota, Ford, and BMW operate within the same technological sector and are direct competitors. As a result, they typically do not share proprietary technologies with one another. Nevertheless, they may acquire advanced components—such as high-performance chips—from firms in other technological domains, such as semiconductor companies like NVIDIA.
For instance, if Toyota licenses a cutting-edge chip design from NVIDIA, it may later transmit derivative knowledge or generalised insights to other automotive firms, such as Ford. However, such diffusion is strictly limited—typically involving high-level performance specifications or integration details, without disclosing the underlying architecture or design. The extent of this knowledge sharing is governed by the original licensing agreement between the licensor (NVIDIA) and the licensee (Toyota), including specific provisions regarding sublicensing. Here a list of technologies possessed by firms:
\begin{description}
    \item \textit{Sector $S_1$: Automotive}
    \begin{description}
        \item Toyota: Electric and hybrid vehicles, advanced driver assistance systems (ADAS).
        \item Ford: Electric engines, autonomous driving software.
        \item BMW: Electric mobility technologies, advanced infotainment interfaces.
    \end{description}
    \item \textit{Sector $S_2$: Semiconductor}
    \begin{description}
        \item NVIDIA: AI chips and autonomous driving processors.
        \item Intel: High-performance processors, cloud computing architectures.
        \item Waymo: Autonomous driving systems based on machine learning.
    \end{description}
    \item \textit{Sector $S_3$: Telecommunication}
    \begin{description}
        \item Verizon: 5G network infrastructure.
        \item AT\&T: Data transmission and mobile communication technologies.
        \item Ericsson: Advanced cellular networks and telecommunications infrastructure.
        \item Huawei: 5G chips, edge computing solutions.
    \end{description}
\end{description}
The interconnections among firms in the technological market are illustrated in Figure~\ref{fig:Iexample}, represented as a graph with \emph{a priori} unions. The set of firms is given by $N = \{1, \ldots, 10\}$, and the partition into technological sectors is $\Pi = \big\{\{1,2,3\}, \{4,5,6\}, \{7,8,9,10\}\big\} \in \Pi_{3}(N)$. Nodes $1$, $2$, and $3$ correspond to the automotive sector $S_1$; nodes $4$, $5$, and $6$ belong to the semiconductor sector $S_2$; and nodes $7$, $8$, $9$, and $10$ form the telecommunication sector $S_3$.

\begin{figure}[htbp]
  	\centering 

\begin{tikzpicture}[scale=0.80]
\Vertex[x=0,y=0,size=.6,color=Green,opacity=.4,label=1]{1}
\node at ($(1) + (-1,-0.3)$) {Toyota};
\Vertex[x=0,y=-2.0,size=.6,color=Green,opacity=.4,label=2]{2}
\node at ($(2) + (-1,0)$) {Ford};
\Vertex[x=0,y=-4.0,size=.6,color=Green,opacity=.4,label=3]{3}
\node at ($(3) + (-1,0.2)$) {BMW};
\Vertex[x=4,y=0.5,size=.6,color=red,opacity=.4,label=4]{4}
\node at ($(4) + (+1.2,-0.4)$) {NVIDIA};
\Vertex[x=4,y=-2.0,size=.6,color=red,opacity=.4,label=5]{5}
\node at ($(5) + (+1,-0.2)$) {Intel};
\Vertex[x=4,y=-4.0,size=.6,color=red,opacity=.4,label=6]{6}
\node at ($(6) + (1,+0.4)$) {Waymo};

\Vertex[x=10,y=0,size=.6,color=blue!,opacity=.4,label=7]{7}
\node at ($(7) + (0,-1)$) {Verizon};
\Vertex[x=8,y=-1.0,size=.6,color=blue,opacity=.4,label=8]{8}
\node at ($(8) + (0,-1)$) {AT\&T};
\Vertex[x=8,y=-3.5,size=.6,color=blue,opacity=.4,label=9]{9}
\node at ($(9) + (0,-1)$) {Ericsson};
\Vertex[x=10,y=-4.0,size=.6,color=blue,opacity=.4,label=10]{10}
\node at ($(10) + (0,+1)$) {Huawei};

\Edge[lw=0.5,color=black,bend=0](1)(2)
\Edge[lw=0.5,color=black,bend=0](1)(4)
\Edge[lw=0.5,color=black,bend=0](1)(5)
\Edge[lw=0.5,color=black,bend=0](1)(6)
\Edge[lw=0.5,color=black,bend=-15](1)(7)
\Edge[lw=0.5,color=black,bend=0](2)(4)
\Edge[lw=0.5,color=black,bend=0](3)(4)
\Edge[lw=0.5,color=black,bend=-30](3)(10)
\Edge[lw=0.5,color=black,bend=0](4)(7)
\Edge[lw=0.5,color=black,bend=0](5)(8)
\Edge[lw=0.5,color=black,bend=0](6)(9)
\Edge[lw=0.5,color=black,bend=0](6)(10)
\Edge[lw=0.5,color=black,bend=0](4)(5)
\Edge[lw=0.5,color=black,bend=0](5)(6)
\node[draw, circle, fit=(1) (2) (3), label=above:$S_1$: Automotive] {};
\node[draw, circle, fit=(4) (5) (6), label=above:$S_2$: Semiconductor] {};
\node[draw, ellipse, fit=(7) (8) (9) (10), label=above:$S_3$: Telecommunication] {};
\end{tikzpicture}  
  \caption{Representation of a technological market.}
        \label{fig:Iexample}
  \end{figure}
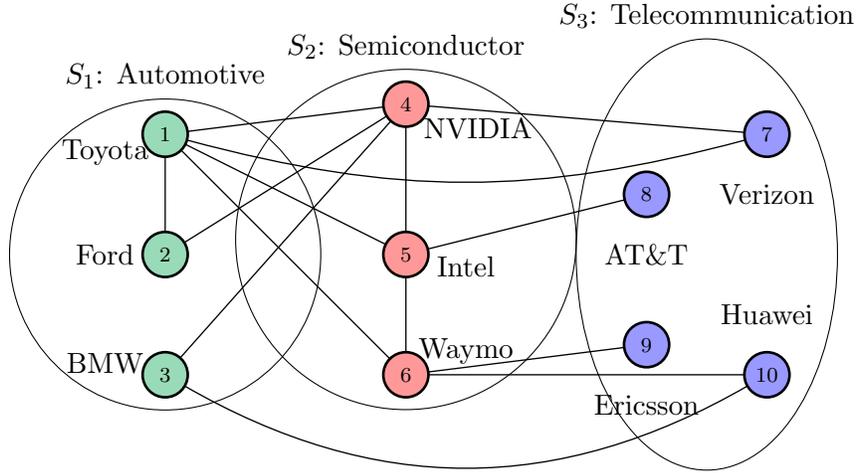
  \end{example}
Understanding  the power of involved firms is crucial for analysing the structure of technology sharing networks. 
A power index on graphs with a priori unions is defined as follows:  
\begin{align*}
\phi: \mathcal{G}(N) \to & \quad \mathbb{R}^n\\ 
(E,\Pi) \to &  \{\phi_i(E,\Pi)\}_{i=1}^{n}
\end{align*} 
We introduce a family of power indices parameterised by two key values. The first, $\alpha \geq 0$, represents the average market value of a primary license—i.e., the standard cost of granting initial access to a proprietary technology. The second, $\beta \geq 0$, reflects the average market value of a sublicense, capturing the value of indirect technological diffusion. Both parameters are treated as exogenous and uniform across firms and sectors, reflecting a competitive and homogeneous market environment. This assumption is particularly reasonable in standardised industries such as telecommunications, software, and pharmaceuticals, where licensing terms tend to converge around market equilibrium. Each index assigns a level of power to a firm based on the number of licenses and sublicenses it issues, weighted respectively by $\alpha$ and $\beta$. The framework thus measures a firm’s influence not only as a direct licensor, but also as a facilitator of technological dissemination throughout the network.\\

\paragraph{Preliminary notation}
Let \( G_r = (N, E) \) a graph with a priori union, then, for each node \( i \in N \), we adopt the following notation:

\begin{itemize}
    \item[-]\( N(i) \subseteq N \): the set of neighbors of node \( i \) in the graph;
    \item[-]\( S^i\in\Pi \): the union to which node \( i \) belongs;
    \item[-]\( S^{-i} := N \setminus S^i \): the set of nodes outside the union of node \( i \);
    \item[-]\( N(i, S) := N(i) \cap S \): the set of neighbors of node \( i \) within a given subset \( S \subseteq N \);
    \item[-]\( d_i(S) := |N(i, S)| \): the number of neighbors of node \( i \) within \( S \);
    \item[-]\( d_i := d_i(N) = |N(i)| \): the total degree of node \( i \).
    
\end{itemize}

Some of the quantities above are particularly relevant for our analysis and are repeatedly used throughout the paper. For clarity and convenience, we introduce dedicated notations:

\begin{itemize}
\item[-]\( d_i^{\mathrm{ext}} := d_i(S^{-i}) = |N(i, S^{-i})| \): the number of neighbors of node \( i \) outside its technological sector or union (also referred to as the external degree of $i$).
    \item[-]\( d_i^{\mathrm{int}} := d_i(S^i) = |N(i, S^i)| \): the number of neighbors of node \( i \) inside its technological sector or union (internal degree);
    \item[-]\( N^{\mathrm{ext}} := \{ i \in N : d_i^{\mathrm{ext}} > 0 \} \): the set of bridge nodes, i.e., nodes that have at least a neighbor in a different union;
    \item[-]\( \epsilon^{\mathrm{ext}} := \left| \left\{ \{i,j\} \in E : S^i \neq S^j \right\} \right| \): the number of links that connect nodes in different unions (also referred to as external links of the technological networks with a priori unions).
\end{itemize}
 
With a slight abuse of notation, we omit the explicit dependence of the above topological quantities on the underlying structure \((E, \Pi)\). However, these quantities inherently depend on such structure. 
For instance, \(N(i\,|\,E,\Pi)\) stands for the neighborhood of node \(i\) under a given edge set and partition, and would normally coincide with \(N(i)\). When the context is clear, we use the simplified form without explicitly specifying the dependence. Nonetheless, in cases where multiple network structures or partitions are being considered, we will write expressions such as \(S^i(E,\Pi)\), \(N(i,S\,|\,E,\Pi)\), or \(d_i(S\,|\,E,\Pi)\) to disambiguate the reference.
This convention allows us to streamline the notation without sacrificing precision where it matters (especially in our proofs).
\begin{definition}
    Given $(\alpha,\beta)\in \mathbb{R}_{\geq 0}^2$, the \emph{$(\alpha,\beta)$-power index} $\varphi^{(\alpha,\beta)} : \mathcal{G}(N) \to \mathbb{R}^n$ is defined by
    \begin{align*}
        \varphi_i^{(\alpha,\beta)}(E,\Pi) & = \alpha\, d_i(S^{-i}) + \beta \sum_{j\in N(i,S^{-i})} \frac{d_j(S^{-i})}{d_j(S^{i})}, \quad \forall (E,\Pi)\in\mathcal{G}(N),\ i\in N.
    \end{align*}
    If $i$ is not a bridge node, the above expression is understood as $\varphi_i^{(\alpha,\beta)}(E,\Pi) = 0$.
    \end{definition}
We allow the parameters $(\alpha,\beta)$ to be non-negative, so that one of the two components can vanish. The case $(\alpha,\beta) = (0,0)$ yields a trivial index equal to zero for all nodes. The \((\alpha,\beta)\)-power index for the ten firms in the technological network \( (E, \Pi) \) of Example~\ref{ex:threeU} is now calculated.
 
\begin{example}[Continuation of Example~\ref{ex:threeU}]
Let \( (\alpha, \beta) \in \mathbb{R}^2_{\geq 0} \). The corresponding \((\alpha,\beta)\)-power index for each firm is summarised in Table~\ref{tab:power_indices_example}.

\begin{table}[H]
\centering
\renewcommand{\arraystretch}{1.4}
\begin{tabular}{|c|c||c|c|}
\hline
Firm & \( \varphi^{(\alpha,\beta)}_i(E,\Pi) \) & Firm & \( \varphi^{(\alpha,\beta)}_i(E,\Pi) \) \\
\hline
1  & \( 4\alpha + \tfrac{23}{3}\beta \) & 6  & \( 3\alpha + \tfrac{5}{3}\beta \) \\
2  & \( \alpha + \tfrac{2}{3}\beta \)   & 7  & \( 2\alpha + 8\beta \) \\
3  & \( 2\alpha + \tfrac{5}{3}\beta \)  & 8  & \( \alpha + 3\beta \) \\
4  & \( 4\alpha + \tfrac{11}{3}\beta \) & 9  & \( \alpha + \beta \) \\
5  & \( 2\alpha + \tfrac{2}{3}\beta \)  & 10 & \( 2\alpha + 2\beta \) \\
\hline
\end{tabular}
\caption{\scriptsize \((\alpha,\beta)\)-power index for the firms in Example~\ref{ex:threeU}.}
\label{tab:power_indices_example}
\end{table}
\end{example}
If market values exhibit a specific symmetry, power indices can be simplified as follows.
\begin{proposition}
\label{prop10}
For any $\alpha\in\mathbb{R}_{\geq 0}$ 
\begin{equation*}
\varphi_i^{(\alpha,\alpha)}(E,\Pi)=\alpha \dsum_{j\in N(i,S^{-i})}\dfrac{d_j}{d_{j}(S^i)} \quad
\forall (E,\Pi)\in \mathcal{G}(N), i\in N,
\end{equation*}

with $\varphi_i^{(\alpha,\alpha)}(E,\Pi)=0$ if $i$ is not a bridge. 
\end{proposition}

The following result examines how power indices respond—either increasing or decreasing—when additional links are added to an existing network with \emph{a priori} unions.
While it is expected that adding a new link between a bridge node \( i \) and a node \( m \) outside its own union increases \( i \)'s power index, an external power effect raises on some of its neighboring nodes. 
To capture this phenomenon, we define a relevant subset of bridge nodes, referred to as \emph{incomplete bridges}. A node \( i \in N \) is said to be \emph{incomplete} if there exists a coalition \( S_\ell \neq S^i \) such that
\[
\emptyset \subsetneq N(i) \cap S_\ell \subsetneq S_\ell, \quad \text{and} \quad N(i) \cap (N \setminus S_\ell) \neq \emptyset.
\]
That is, node \( i \) is partially connected to the coalition \( S_\ell \), and it also has neighbors outside \( S_\ell \). 
In Figure~\ref{fig:incomplete_bridge_extension}(a), node \( 1 \) is an incomplete bridge and is a neighbor of node \( 4 \) in another union. In Figure~\ref{fig:incomplete_bridge_extension}(b), the addition of a second link between nodes \( 1 \) and \( 6 \) increases the power of node \( 1 \), while potentially diminishing the relative power of node \(4\).

\begin{figure}[htbp]
\centering
\begin{minipage}[b]{0.45\textwidth}
\centering
\textbf{(a)} Initial graph $(E,\Pi)$
\begin{tikzpicture}[scale=0.5, every node/.style={circle, draw, minimum size=0.3cm, font=\small}]
\node[fill=blue!20] (1) at (0,3) {1};
\node[fill=blue!20] (2) at (-2,3) {2};
\node[fill=blue!20] (3) at (2,3) {3};
\node[fill=green!20] (4) at (0,0) {4};
\node[fill=green!20] (5) at (2,0) {5};
\node[fill=green!20] (6) at (4,0) {6};

\draw[thick] (1) -- (2);
\draw[thick] (1) -- (3);
\draw[thick] (1) -- (4);
\draw[thick] (4) -- (5);
\draw[thick] (5) -- (6);
\node[draw, ellipse, fit=(1) (2) (3) , label=right:$S_1$] {};
\node[draw, ellipse, fit=(4) (5) (6) , label=right:$S_2$] {};

\end{tikzpicture}
\end{minipage}
\begin{minipage}[b]{0.45\textwidth}
\centering
\textbf{(b)} Extended graph $(E',\Pi)$
\begin{tikzpicture}[scale=0.5, every node/.style={circle, draw, minimum size=0.6cm, font=\small}]
\node[fill=blue!20] (1) at (0,3) {1};
\node[fill=blue!20] (2) at (-2,3) {2};
\node[fill=blue!20] (3) at (2,3) {3};
\node[fill=green!20] (4) at (0,0) {4};
\node[fill=green!20] (5) at (2,0) {5};
\node[fill=green!20] (6) at (4,0) {6};

\draw[thick] (1) -- (2);
\draw[thick] (1) -- (3);
\draw[thick] (1) -- (4);
\draw[thick] (4) -- (5);
\draw[thick] (5) -- (6);
\draw[thick, red] (1) to (6); 

\node[draw, ellipse, fit=(1) (2) (3) , label=right:$S_1$] {};
\node[draw, ellipse, fit=(4) (5) (6) , label=right:$S_2$] {};
\end{tikzpicture}
\end{minipage}

\vspace{0.5cm}
\caption{In the initial network \( (E, \Pi) \), node \( 1 \) (in blue, part of \( S_1 = \{1,2,3\} \)) has a unique neighbor \( 4 \) in the green union \( S_2 = \{4,5,6\} \), making it an incomplete bridge.
In the extended network \( (E', \Pi) \), a second link is added from \( 1 \) to node \( 6 \).}
\label{fig:incomplete_bridge_extension}
\end{figure}
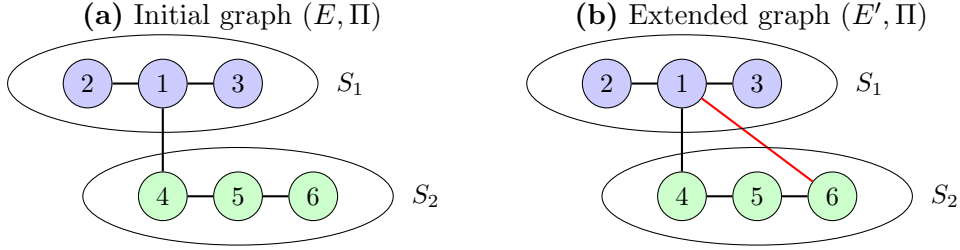

The term \emph{incomplete} indicates that the neighborhood of \( i \) in \( S_\ell \) is non-empty but not exhaustive.
In the context of technological networks, the presence of incomplete bridge nodes signals underlying structural informational asymmetries. Specifically, firms operating within the same technological sector may not have equal access to available technologies, as some connections are selectively mediated through bridge firms. This asymmetry is particularly significant in networks with a priori unions, where it is closely linked to the anti-monotonicity property of the power index discussed earlier. We formalize these consideration in the following result.
\begin{proposition} \label{monot}
Suppose that \( (E, \Pi) \) contains an incomplete bridge node. 
There exists an extended set of edges \( E' \supsetneq E \) for which:
\begin{itemize}
\item[-] there exists a node \( i \in N \) whose power index strictly increases, i.e., \( \varphi_i^{(\alpha,\beta)}(E', \Pi) > \varphi_i^{(\alpha,\beta)}(E, \Pi) \)
with the additional requirement that \( \alpha > 0 \);
\item[-] 
there exists a node \( j \in N \) whose power index strictly decreases, i.e.,    
$\varphi_j^{(\alpha,\beta)}(E', \Pi) < \varphi_j^{(\alpha,\beta)}(E, \Pi)$,
with the additional requirement that \( \beta > 0 \).
\end{itemize}
\end{proposition}

Constructing rankings of firms in this context helps to evaluate the market concentration of technological markets. Firms' rankings highlight which firms exert the greatest power—either through direct licensing or via indirect control mechanisms such as sub-licensing—and thus help to identify the key actors in complex technological landscapes. Firms in \( N \) are ranked according to an ordering induced by the power index, which reflects their relative power within  technological networks. Given a network configuration \( (E, \Pi) \in \mathcal{G}(N) \) and $(\alpha,\beta)\in\mathbb{R}_{\geq0}^2$, the ranking order \( \mathcal{R}^0 \) on $N$ is defined as follows: \( (i,j)\in \mathcal{R}^0\), meaning that firm \( i \) is at least as important than  firm \( j \), if and only if
$\varphi_i^{(\alpha,\beta)}(E, \Pi) \geq \varphi_j^{(\alpha,\beta)}(E, \Pi)$; moreover we say that $i$ precedes $j$ if $(i,j)\in  \mathcal{R}^0$ and $(j,i)\notin  \mathcal{R}^0$. Firms are thus ordered in decreasing order of power according to \( \mathcal{R}^0 \). Ranks are assigned as integers.
We expect that, if a node \( i \) exhibits both stronger direct licensing strength and greater sublicensing potential than node \( j \), then \( i \) will have a higher power index—and thus a higher rank—than \( j \), independently of the specific values of the parameters \( \alpha \) and \( \beta \). More intriguing is the scenario in which one node has stronger direct licensing strength, while the other has greater sublicensing potential. In such cases, the comparison between the two nodes—in terms of both power index and resulting rank—depends on the values of \( \alpha \) and \( \beta \), and more specifically on their ratio \( \rho = \frac{\beta}{\alpha} \).
The following proposition clarifies this insight.

\begin{proposition} \label{thm change}
Let $(E,\Pi)\in\mathcal{G}(N)$ and suppose that there are \( i, j \in N \) such that

\begin{equation} \label{assumption change}
\Big( d_i(S^{-i}) -d_j(S^{-j}) \Big)
\left(\sum_{h \in N(i, S^{-i})} \frac{d_h(S^{-i})}{d_h(S^{i})}-  
\sum_{h \in N(j, S^{-j})} \frac{d_h(S^{-j})}{d_h(S^{j})}
\right)<0.
\end{equation}
Let \( \rho = \frac{\beta}{\alpha} \), with $\alpha>0$.  Then there exists a threshold \( \rho^{*} = \rho^{*}(i, j) > 0 \) such that one of the two alternative scenarios holds:
\begin{itemize}
    \item[-] firm \( i \) precedes \( j \) in \( \mathcal{R}^{0} \) whenever \( \rho > \rho^{*} \), and firm \( j \) precedes firm \( i \) when \( \rho < \rho^{*} \);
    \item[-] firm \( i \) precedes firm \( j \) in \( \mathcal{R}^{0} \) whenever \( \rho < \rho^{*} \), and firm \( j \) precedes firm \( i \) when \( \rho > \rho^{*} \).
\end{itemize}
\end{proposition}

Proposition~\ref{thm change} shows that, for any pair of firms \( i \) and \( j \), there may exist a critical value \( \rho^*(i, j) \) at which the two firms exchange their ranks. As a result, multiple rank reversals may occur across the rankings as \( \rho \) varies.  This phenomenon is illustrated in Table~\ref{ranking esempio1 r0}, which reports the firm rankings from Example~\ref{ex:threeU} for a range of values of \( \rho \) between 0 and 1. Condition \( \rho \leq 1 \) corresponds to the situation where \( \alpha>0 \) is greater than \( \beta \geq 0 \), implying that the market value of the licensed technology is higher than the sublicensing usage of the same technology.

\begin{table}[H]
\centering
\resizebox{\textwidth}{!}{%
\begin{tabular}{lccccccccccc}
\toprule
Rank & $\rho = 0$ & $0 < \rho < \frac{3}{19}$ & $\rho = \frac{3}{19} \approx 0.16$ & $\frac{3}{19} < \rho < \frac{3}{7}$ & $\rho = \frac{3}{7} \approx 0.43$ & $\frac{3}{7} < \rho < \frac{6}{13}$ & $\rho = \frac{6}{13} \approx 0.46$ & $\frac{6}{13} < \rho < \frac{3}{4}$ & $\rho = \frac{3}{4} = 0.75$ & $\frac{3}{4} < \rho < 1$ & $\rho = 1$ \\
\midrule
Rank 1 & 1, 4 & 1 & 1 & 1 & 1 & 1 & 1 & 1 & 1 & 1 & 1 \\
Rank 2 & 6 & 4 & 4 & 4 & 4 & 4, 7 & 7 & 7 & 7 & 7 & 7 \\
Rank 3 & 3, 5, 7, 10 & 6, 7 & 7 & 7 & 7 & 6 & 4 & 4 & 4 & 4 & 4 \\
Rank 4 & 2, 8, 9 & 10 & 6 & 6 & 6 & 10 & 6 & 6 & 6 & 6 & 6 \\
Rank 5 &  & 3 & 10 & 10 & 10 & 3 & 10 & 10 & 10 & 8, 10 & 8, 10 \\
Rank 6 &  & 5 & 3 & 3 & 3 & 8 & 3 & 3, 8 & 8 & 3 & 3 \\
Rank 7 &  & 8 & 5 & 5, 8 & 8 & 5 & 8 & 5 & 3 & 5 & 5 \\
Rank 8 &  & 9 & 8 & 9 & 5 & 9 & 5 & 9 & 5 & 9 & 9 \\
Rank 9 &  & 2 & 9 & 2 & 9 & 2 & 9 & 2 & 9 & 2 & 2 \\
Rank 10 &  &  & 2 &  & 2 &  & 2 &  & 2 &  &  \\
\bottomrule
\end{tabular}
}
\caption{Node rankings across multiple intervals of the value $\rho = \beta/\alpha$ subject to $ \rho \leq 1$.}
\label{ranking esempio1 r0} 
\end{table}
Rankings evolve as the parameter \( \rho = \frac{\beta}{\alpha} \) increases, reflecting a shift in the economic balance between sub-licensing and direct licensing. A higher value of \( \rho \) means that sub-licenses become more valuable relative to direct licenses, increasing the weight of indirect technological influence in the network and thus altering the ranking of firms. Specifically:
Firm~1 consistently holds the top rank for all values of \( \rho \), highlighting its dominant position regardless of the relative importance of sub-licensing. Firm~4 is ranked second when \( 0 \leq \rho \leq \frac{6}{13} \), but falls to third place once \( \rho > \frac{6}{13} \), reflecting its reliance on direct licensing strength.
Firm~7, despite a lower degree, overtakes firm~4 when \( \rho > \frac{6}{13} \), showing that higher sub-licensing value enhances indirect power and improves its position in the ranking.

We now examine whether firms can be prevented from altering their market-ranks under specific topologies of technological networks. To this end, we identify sector-specific configurations in which firms interact exclusively within a single technological sector. Under such conditions, cross-sector technological spillovers are structurally constrained, as formalized below.
\begin{definition}
A graph with \emph{a priori} unions \( (E, \Pi) \) is said to satisfy the \emph{restricted spillover property} if, for each firm \( i \in N \), there exists a unique index \( \psi(i) \in \{1, \dots, r\} \) such that \( N(i) \subseteq S_{\psi(i)} \), with \( S_{\psi(i)} \neq S^i \).
\end{definition}
When this structural assumption holds, each firm may interact with firms outside its own sector, but only within a single, distinct technological union. As a result, cross-sector interactions are strictly limited, and firms are unable to share technologies from multiple technological sectors. Consequently, market power is more limited. In particular, rankings remain fixed for all values of \( \alpha \) and \( \beta \), as stated in the following result.
\begin{proposition}
\label{prop29}
Let \( (E, \Pi) \) be a graph with a priori unions that satisfies the restricted spillover property. Then, the rankings induced by \( \mathcal{R}^{0} \) remain invariant under any variation in the market parameters \( \alpha \) and \( \beta \), with $\alpha\neq 0$.
\end{proposition}

An illustrative example of a network satisfying the restricted spillover property is provided in Figure~\ref{net13}.
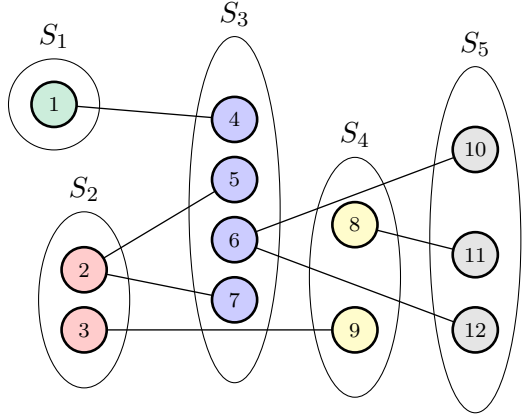
\begin{figure}[htbp]
\centering
\begin{tikzpicture}[scale=0.4]
\Vertex[size=0.6,color=Green,opacity=.2,label={1}]{1}
\Vertex[x=1,y=-5.5,size=0.6,color=red,opacity=.2,label={2}]{2}
\Vertex[x=1,y=-7.5,size=0.6,color=red,opacity=.2,label={3}]{3}
\Vertex[x=6,y=-0.5,size=0.6,color=blue,opacity=.2,label={4}]{4}
\Vertex[x=6,y=-2.5,size=0.6,color=blue,opacity=.2,label={5}]{5}
\Vertex[x=6,y=-4.5,size=0.6,color=blue,opacity=.2,label={6}]{6}
\Vertex[x=6,y=-6.5,size=0.6,color=blue,opacity=.2,label={7}]{7}
\Vertex[x=10,y=-4.0,size=0.6,color=yellow,opacity=.2,label={8}]{8}
\Vertex[x=10,y=-7.5,size=0.6,color=yellow,opacity=.2,label={9}]{9}
\Vertex[x=14,y=-1.5,size=0.6,color=gray,opacity=.2,label={10}]{10}
\Vertex[x=14,y=-5.0,size=0.6,color=gray,opacity=.2,label={11}]{11}
\Vertex[x=14,y=-7.5,size=0.6,color=gray,opacity=.2,label={12}]{12}
\Edge[lw=0.5,color=black,bend=0](1)(4)
\Edge[lw=0.5,color=black,bend=0](2)(5)
\Edge[lw=0.5,color=black,bend=0](2)(7)
\Edge[lw=0.5,color=black,bend=0](3)(9)
\Edge[lw=0.5,color=black,bend=0](6)(10)
\Edge[lw=0.5,color=black,bend=0](6)(12)
\Edge[lw=0.5,color=black,bend=0](8)(11)

\node[draw, ellipse, fit=(1) , label=above:$S_1$] {};
\node[draw, ellipse, fit=(2) (3) , label=above:$S_2$] {};
\node[draw, ellipse, fit=(4) (5) (6) (7), label=above:$S_3$] {};
\node[draw, ellipse, fit=(8) (9), label=above:$S_4$] {};
\node[draw, ellipse, fit=(10) (11) (12), label=above:$S_5$] {};
\end{tikzpicture}
\caption{A network with \emph{a priori} unions and restricted spillovers.}
\label{net13}
\end{figure}

\section{On the cooperative representation of indices}\label{coope}
In Section~\ref{sect1}, the power index \( \varphi^{(\alpha,\beta)} \) assigned to a node (firm) quantifies not only its licensing power but also its capacity to influence the network through licenses and sub-licenses. In this section, we explore whether this index remains a reliable measure in a cooperative setting, where firms engage in joint licensing activities. In such contexts, the index should capture the stable distribution of power arising from technology exchange within the network.

Before proceeding with this analysis, we introduce some key definitions from TU-games. 
When firms collaborate, these indices reflect a stable distribution of power, resulting from the sharing of technologies across the network.
Recall that a \emph{TU-game} is defined as a pair \( (N, \nu) \), where \( N \in \mathbb{N} \) is the set of players, with \( |N| = n \), and \( \nu : 2^N \to \mathbb{R} \) is the \emph{characteristic function}. This function assigns a real worth to each coalition and satisfies the standard normalisation \( \nu(\emptyset) = 0 \). The set \( 2^N \) denotes the power set of \( N \), and any subset \( T \subseteq N \) is referred to as a \emph{coalition}, with \( \nu(T) \) representing its \emph{worth}.

We associate a TU-game to each $(\alpha,\beta)\in\mathbb{R}_{\geq0}^2$ and each graph 
 \( (E, \Pi) \in \mathcal{G}(N) \) , where \( \Pi \in \mathbf{\Pi}_r(N) \) is a partition of firms into \( r \) distinct technological sectors. The characteristic function of the game captures both the coalition's direct strength to provide technology and its potential to facilitate its diffusion through the network. 
The value of each coalition should reflect both the individual strength of its nodes and their capacity to facilitate diffusion.

We begin by specifying the worth of a coalition consisting entirely of firms from the same technological sector. Let \( T \subseteq S_\ell \), where \( S_\ell \in \Pi \) and \( \ell \in \{1, \dots, r\} \). The worth \( \nu(T) \) is then defined as a measure of the coalition’s strength to reach firms outside its own sector—either directly or through sublicensing connections.
The one-step worth of a coalition \(T\), denoted by \(\nu_1(T)\), quantifies the number of firms (nodes) that can receive the technology via direct licensing from members of \(T\), excluding firms already within the same union. Formally:
\[
\nu_1(T) := \alpha \sum_{i \in T} d_i(N\setminus S_\ell),
\]
where 
$\alpha\geq0$  represents the average value of a direct license in that technological market.
The two-step worth, \(\nu_2(T)\), accounts for the potential for further diffusion through sublicensing. It measures the number of firms that can be reached via neighbors of \(T\) located outside \(S_\ell\), again excluding those already within the union. Formally:
\[
\nu_2(T) := \beta \sum_{h \in N(T) \cap (N \setminus S_\ell)} d_h(N \setminus S_\ell),
\]
where \(N(T)\) is the set of neighbors of \(T\). The parameter 
$\beta\ge 0$  represents the average value of a sublicense - i.e., a secondary transfer of technology mediated by firms that initially received it from the coalition \(T\).
When a coalition \(T\) includes firms from multiple technological sectors, its worth is computed separately for each sector.
Indeed, this is coherent with our initial assumption that sharing proprietary technologies across different markets is not a loss. No market externalities are assumed in the technological network.
 Consequently, the total worth of 
\(T\) is the sum of its contributions within each union.
We define the characteristic function \(\nu: 2^N \to \mathbb{R}\) as follows:

\begin{equation}\label{eq:charactfunc}
	\begin{aligned}
	   &\nu(\emptyset) := 0,\\
	   &\nu(T) := \sum_{\ell=1}^{r} \nu_1(T \cap S_\ell)+\nu_2(T \cap S_\ell)\footnote{Note that if $T\cap S_\ell=\emptyset$ then $\nu_1(\emptyset)=\nu_2(\emptyset)=0$.},
	\end{aligned} 
\end{equation} 
where \( r \) is the number of distinct technological markets, and \( T \cap S_\ell \) represents the subset of firms in \( T \) operating in the technological market \( S_\ell \). 
Let \( \mathcal{U}(N) \) denote the set of all transferable utility (TU) games defined on set \( N \).
The central definition follows:
\begin{definition}
For any graph with a priori unions, $(E,\Pi)\in\mathcal{G}(N)$, and average market values , $(\alpha,\beta)\in\mathbb{R}_{\geq0}^2$, let \( (N, \nu; (E, \Pi), \alpha, \beta)\) be the associated TU-game with characteristic function $\nu$ defined in \eqref{eq:charactfunc}.\footnote{For brevity, and when the context is clear, we will sometimes denote the game simply as 
$(N,\nu)$, leaving implicit the network structure, the partition, and other coefficients.}
\end{definition}

We study some basic properties of the proposed TU game as follows.
\begin{proposition}\label{cic}
Let \( (N, \nu; (E, \Pi), \alpha, \beta) \) be the proposed TU-game. All the following sentences are true:
\begin{enumerate}
\item
The game is subadditive.
\item
If \( \beta = 0 \) then the game is additive.
\item
Suppose that \( \Pi \) consists entirely of singleton unions then the game is additive.
\item
Suppose $\beta >0$ and that there exist two bridge nodes \( i, j \in S_\ell \) within the same union \( S_\ell \) such that
$
N(i, N \setminus S_\ell) \cap N(j, N \setminus S_\ell) \neq \emptyset,$
then the game is non-additive.
\end{enumerate}
\end{proposition}
Under the hypothesis of point 4 of the previous proposition the following condition then holds:
\begin{equation}\label{stadd}
\nu(T \cup S) < \nu(T) + \nu(S) \quad \mbox{for some disjoint coalitions} \quad T, S \subseteq N.
\end{equation}
In such cases, one of the \textit{most desirable solution concepts}—the core— fails to provide any viable allocation outcome. We, then, adopt the Shapley value as a solution concept to provide foundations of stability to our power indices.
The Shapley value is a function \( \mathrm{Sh} : \mathcal{U}(N) \to \mathbb{R}^n \) that assigns to each TU-game \( (N, \nu) \in \mathcal{U}(N) \) a vector of payoffs. For each player \( i \in N \), it is given by:
\[
\mathrm{Sh}_i(N, \nu) = \sum_{T \subseteq N \setminus \{i\}} \frac{|T|!\, (n - |T| - 1)!}{n!} \Big( \nu(T \cup \{i\}) - \nu(T) \Big).
\]
It is well known that computing the Shapley value can be computationally demanding.
If a game is additive, the Shapley value of node $i$ coincides with the value of the characteristic function in the singleton $\{i\}$. Therefore, the computation of the Shapley value is interesting only in the case of non additive games (see Proposition \ref{cic}).  We will use the approach based on the Harsanyi dividends to calculate the Shapley value (see \citep{harsanyi1959bargaining, billot2005share}). 
The computation of the Shapley value can be greatly simplified if one can calculate Harsanyi dividends. This approach is particularly convenient for TU-games that have a relatively poor decomposition in the basis of unanimity games.

\paragraph{The  calculation of the Shapley value via Harsanyi dividends}
Let $T$ a nonempty subset of $N$, $(N,u_T)$ is a unanimity game if $u_T(S)=1$, if $T\subseteq S$, and $0$, otherwise. Every TU-game $(N, \nu) \in \mathcal{U}(N)$ admits a unique decomposition as a linear combination of unanimity games:
\[
\nu = \sum_{\emptyset \neq T \subseteq N} \Delta_{\nu}(T)\, u_T,
\]
where the coefficients $\Delta_{\nu}(T)$, known as \emph{Harsanyi dividends}, are given by
\[
\Delta_{\nu}(T) = \sum_{\emptyset \neq S \subseteq T} (-1)^{|T| - |S|} \nu(S),
\]
or equivalently, via the recursive formula: $\Delta_{\nu}(\emptyset) = 0$ and $\Delta_{\nu}(T) = \nu(T) - \sum_{S \subsetneq T} \Delta_{\nu}(S)$ (see~\citep{shapley1953b}).
This decomposition allows the Shapley value of $\nu$ to be written as a linear combination of the Shapley values of unanimity games. For each node $i \in N$ it holds that
\[
\mathrm{Sh}_i(\nu) = \sum_{T \ni i} \frac{\Delta_{\nu}(T)}{|T|},
\]
since $\mathrm{Sh}_i(u_T) = 1/|T|$ when $i \in T$ and $0$ otherwise.
The next result shows that for any game \( (N, \nu) \), all Harsanyi dividends associated with coalitions not entirely contained within a single technological sector vanish, greatly reducing the computational burden.
\begin{proposition}
\label{prop9}
Let \( (N, \nu; (E, \Pi), \alpha, \beta) \) be the proposed TU-game, and let \( T \subseteq N \) be a coalition not entirely contained in any union \( S_\ell \), \( \ell \in \{1, \ldots, r\} \). Then,
$
\Delta_{\nu}(T) = 0.
$
\end{proposition}
Hence, the only nontrivial Harsanyi dividends correspond to coalitions \( T \subseteq S_\ell \) for some \( \ell \). Exploiting the linearity of dividends, we observe:
$
\Delta_{\nu_1 + \nu_2}(T) = \Delta_{\nu_1}(T) + \Delta_{\nu_2}(T)$ for any $T \subseteq S_\ell.
$
Since \( (N, \nu_1) \) is additive, \( \Delta_{\nu_1}(T) = 0 \) for all \( |T| \geq 2 \), and for singleton coalitions:
$
\Delta_{\nu_1}(\{i\}) = \nu_1(\{i\}).$
Indeed, the challenge is whether non-zero Harsanyi dividends can be calculated in  a relatively simple way.
Unfortunately, due to their recursive and combinatorial nature, Harsanyi dividends grow more complex as the number of nodes with overlapping neighborhoods increases. 
However, in Table~\ref{dividends}, we report the nonzero Harsanyi dividends for the network with unions in Example~\ref{ex:threeU}. Notably, only 19 out of the 1023 possible coalitions have nonzero dividends.
\begin{table}[H]
\centering
\scalebox{1.0}{

\begin{tabular}{|c|r|c|r|c|r|}
\hline
$T$ & $\Delta_\nu(T)$ & $T$ & $\Delta_\nu(T)$ & $T$ & $\Delta_\nu(T)$ \\
\hline
$\{1\}$ & $4\alpha + 9\beta$ & $\{8\}$ & $\alpha + 3\beta$ & $\{4,6\}$ & $-2\beta$ \\
\hline
$\{2\}$ & $\alpha + 2\beta$ & $\{9\}$ & $\alpha + 2\beta$ & $\{5,6\}$ & $-2\beta$ \\
\hline
$\{3\}$ & $2\alpha + 3\beta$ & $\{10\}$ & $2\alpha + 3\beta$ &  $\{9,10\}$ & $-2\beta$ \\
\hline
$\{4\}$ & $4\alpha + 5\beta$ & $\{1,2\}$ & $-2\beta$ & $\{1,2,3\}$ & $2\beta$ \\
\hline
$\{5\}$ & $2\alpha + 2\beta$ & $\{1,3\}$ & $-2\beta$ & $\{4,5,6\}$ & $2\beta$  \\
\hline
$\{6\}$ & $3\alpha + 3\beta$ & $\{2,3\}$ & $-2\beta$ &   & \\
\hline
$\{7\}$ & $2\alpha + 8\beta$ & $\{4,5\}$ & $-2\beta$ & &\\
\hline
\end{tabular}
}
\caption{\scriptsize Nonzero Harsanyi dividends for the cooperative game associated with Example~\ref{ex:threeU}.}
\label{dividends}
\end{table}
Using the dividends listed in Table~\ref{dividends}, we can compute the Shapley value for each firm. For example, firm~1 appears in four nonzero coalitions.
By reading the values of
$\Delta_\nu(\{1\})$,  $\Delta_\nu(\{1,2\})$,  $\Delta_\nu(\{1,3\})$,  $\Delta_\nu(\{1,2,3\})$
we calculate the Shapley value of firm $1$ as follows:
\[
\Delta_\nu(\{1\}) + \frac{\Delta_\nu(\{1,2\})}{2} + \frac{\Delta_\nu(\{1,3\})}{2} + \frac{\Delta_\nu(\{1,2,3\})}{3} = 4\alpha + \frac{23}{3}\beta = \varphi_1^{(\alpha,\beta)}.
\]
For firm~7, only one nonzero dividend appears. Reading the value of
$\Delta_\nu(\{7\}) = 2\alpha + 8\beta$
we obtain that the Shapley value of node 7 is
$\varphi_7^{(\alpha,\beta)}.$
We leave to the reader the calculation of the solution for the remaining firms as an exercise.

Due to the recursive definition of Harsanyi dividends and the computational burden that arises with a high distribution of external edges, we choose to compute the Shapley value of the proposed TU game through its potential function, as detailed in the following subsection.

\paragraph{The calculation of the Shapley value via the potential of the Shapley value}
As shown in \citep{hart1988potential}, the Shapley value can be interpreted as a vector of marginal contributions derived from a potential function defined on the space of TU-games, as follows.
Let \( \Gamma = \{ \mathcal{U}(N) : N \in \mathbb{N} \} \) denote the class of all TU-games. A function \( \mathcal{P} : \Gamma \to \mathbb{R} \) is said to define the \emph{marginal contribution} of player \( i \in N \) in a game \( (N, \nu) \) as
\[
D_i \mathcal{P}(N, \nu) = \mathcal{P}(N, \nu) - \mathcal{P}(N \setminus \{i\}, \nu),
\]
where \( (N \setminus \{i\}, \nu) \) is the subgame obtained by restricting \( \nu \) to coalitions in \( 2^{N \setminus \{i\}} \).
A function \( \mathcal{P} : \Gamma \to \mathbb{R} \) is called a \emph{potential function} if it satisfies the normalization \( \mathcal{P}(\emptyset, \nu) = 0 \) and the identity:
\begin{equation} \label{sumdi}
\sum_{i \in N} D_i \mathcal{P}(N, \nu) = \nu(N),
\end{equation}
for all TU-games \( (N, \nu) \). This property implies that the sum of the marginal contributions of all players equals the total value of the grand coalition.
The following result characterizes the Shapley value via such a potential function:
\begin{theorem}[Theorem A, p.~129, \citep{hart1988potential}]
\label{mainre}
There exists a unique potential function \( \mathcal{P} \) such that, for every TU-game \( (N, \nu) \), the vector of marginal contributions \( \left( D_i \mathcal{P}(N, \nu) \right)_{i \in N} \) coincides with the Shapley value of the game. Moreover, this function is uniquely determined by equation~\eqref{sumdi}, applied to the game and all its subgames \( (T, \nu) \) for \( T \subseteq N \).
\end{theorem}
We now construct an explicit function \( P : 2^N \to \mathbb{R} \) corresponding to the potential of the TU-game defined in the previous section. This formulation will allow us to derive a closed-form expression for the Shapley value.
\begin{definition} \label{potential}
Let \( (N, \nu; (E, \Pi), \alpha, \beta) \) be the proposed TU-game, where \( \Pi = \{S_1, \dots, S_r\} \) with \( r \geq 2 \). Define the map \( P : 2^N \to \mathbb{R} \) as follows:
\begin{itemize}
    \item[-] \( P(\emptyset) = 0 \);
    \item[-] If \( T \subseteq S_\ell \) for some \( \ell \in \{1, \dots, r\} \), then
    \begin{equation} \label{kernel}
    P(T) = \alpha \sum_{j \in T} d_j(S^{-j}) + \beta \sum_{h \in N \setminus S_\ell} \lambda_h(T),
    \end{equation}
    where
    \begin{equation} \label{kernel1}
    \lambda_h(T) = 
    \begin{cases}
     \dsum_{t=1}^{d_h(T)} \dfrac{1}{t}\ d_h(N \setminus S_\ell)  & \text{if } d_h(T) \geq 1, \\
    0 & \text{if } d_h(T) = 0.
    \end{cases}
    \end{equation}
    \item[-] For coalition \( T \subseteq N \), define
    \begin{equation} \label{kernel2}
    P(T) = \sum_{\ell = 1}^{r} P(T \cap S_\ell).
    \end{equation}
\end{itemize}
\end{definition}
The next result shows that this function provides an exact expression for the potential of the TU-game.
\begin{proposition} \label{pot}
For any TU-game \( (N, \nu; (E, \Pi), \alpha, \beta) \) associated to $(E, \Pi)\in\mathcal{G}(N)$ and $(\alpha, \beta)\in\mathbb{R}_{\geq0}^2$, the value of the potential function \( \mathcal{P}(N, \nu) \) of Theorem \ref{mainre} coincides with \( P(N) \), where \( P \) is the function defined in Definition~\ref{potential}.
\end{proposition}
As a direct consequence, we obtain one of the main results of this paper that returns the Shapley value as equal to the power index.
\begin{theorem} \label{thmshapley}
Let \( (N, \nu; (E, \Pi), \alpha, \beta) \) be the proposed TU-game, and let \( i \in N \). Then, the Shapley value of node \( i \) is given by
\[
\mathrm{Sh}_i(N, \nu; (E, \Pi), \alpha, \beta)= \varphi_i^{(\alpha,\beta)}(E, \Pi).
\]
\end{theorem}
The sum of the Shapley values across all firms (or nodes) in technological networks provides a comprehensive measure of the overall power in networks with \emph{a priori} unions, that is, $\nu(N)$. More specifically, the result is the following.

\begin{proposition}
\label{totalpower}
For any network configuration \( (E, \Pi) \in \mathcal{G}(N) \) and any pair of parameters \( (\alpha, \beta) \), we obtain the following:
\begin{equation} \label{eff} 
\sum_{i=1}^n \varphi_i^{(\alpha,\beta)}(E, \Pi)= 2 \alpha\,\epsilon^{\text{ext}} + \beta \sum_{i \in N} \sum_{\ell \in u(i)} d_i(N \setminus S_\ell),
\end{equation}
where the set 
\[
u(i)  = \left\{ \ell \in \{1, \ldots, r\} : d_i(S_\ell) \neq 0 \text{ and } S_\ell \neq S^i \right\}.
\]
\end{proposition}

This efficient value — that is, $\nu(N)$ — or the sum of power indices across all firms is not sum-constant across the set of networks with \emph{a priori} unions.  

However, the efficient value encompasses two terms on the right-hand side of equality~\eqref{eff}. The first term measures the power $\epsilon^{\text{ext}}$ due to the external edges in technological networks with \emph{a priori} unions. The second term can be interpreted as a form of \textit{shadow} power associated with \textit{potential} external links in networks. Specifically, these \textit{shadow} novel external edges are the ones visible if firms unilaterally transition from one technological sector to another union; and such a move is allowed only if they are connected to the union in the original network.

\section{On the axiomatic representation of  indices}\label{Axioms}

In Section 2, we derived the power index under the assumption that firms actively participate in licensing and sub-licensing agreements, as structured within the network we described. In Section 3, we considered a scenario where firms engage in deterministic mutual exchanges of technologies and licenses, influenced by the effects of cooperation. 

Next we take a more abstract approach: rather than assuming that any exchange is predetermined, we investigate whether the proposed power index can be characterized axiomatically, considering the inherent uncertainty in firms' interactions. In this section only, we do not restrict $\alpha$ and $\beta$ to be positive; instead, we allow them to take arbitrary real values. Accordingly, our results hold for any $\alpha, \beta \in \mathbb{R}$.
We focus on identifying the fundamental principles that a power measure should satisfy. This allows us to determine whether our index naturally arises from a set of minimal, intuitive conditions that govern technology sharing within integrated markets. In the following, we present a series of axioms and briefly discuss their economic interpretation within technological networks.
\begin{axiom}[Edge nullity (EN)]\label{Ax-en}
A power index $\phi$ satisfies the \emph{Edge nullity} property if, for every partition 
$\Pi \in \mathbf{\Pi}_r(N)$ with $r \leq n$, it holds that:
\begin{equation*}
    \phi(\emptyset, \Pi) = \underline{0}.
\end{equation*}
\end{axiom}

This axiom states that if there are no edges in the technological network (i.e., the set $E$ is empty), then the power of firms is simply zero.

\begin{axiom}[Inter-Union Neighborhood (IUN)]\label{Ax-ld}
A power index $\phi$ satisfies the \emph{Inter-Union Neighborhood} (IUN) property if, for every $(E,\Pi) \in \mathcal{G}(N)$, it holds that:
\begin{equation*}
    \phi_i(E, \Pi) = \phi_i(E_i^*, \Pi),
\end{equation*}
where $E_i^* = \left\{ \{j, k\} \in E : j \in N(i, S^{-i}|E,\Pi), k \in N \right\}$.
\end{axiom}
The power of a  firm  is determined by the relationships it maintains with its neighbors, with particular emphasis on its ability to connect across different unions or technological sectors. The term Inter-Union Neighborhood reflects this notion, asserting that a firm's influence is contingent upon its capacity to interact with firms in different technological markets. If a firm has no edges with nodes in distinct unions, its overall influence is null. 

\begin{axiom}[Anonimity (A)]\label{Ax-a}
A power index $\phi$ satisfies \emph{Anonimity} property if, for any permutation $\sigma: N \to N$ and for all $(E, \Pi) \in \mathcal{G}(N)$, the following condition holds:
\begin{equation*}
    \phi_i(E, \Pi) = \phi_{\sigma i}(\sigma E, \sigma \Pi),
\end{equation*}
where $\sigma E := \left\{ \{\sigma i, \sigma j\} : \{i, j\} \in E \right\}$ and $\sigma \Pi := (\sigma S_1, \ldots, \sigma S_r)$, with for any $A \subseteq N$, $\sigma A := \{\sigma a : a \in A\}$.
\end{axiom}
Anonimity axiom expresses that the power of a firm is independent of its label or identity. In other words, the influence of a firm, within the technological network, is determined solely by its position and its relationships with other firms, not by which specific node it is. Under this axiom, if a permutation of the set of nodes is applied to both the set of edges and the coalition structure, the power index of each firm remains unchanged.

\begin{axiom}[Linearity (Li)]\label{Ax-l}
A power index $\phi$ satisfies the \emph{Linearity} property if, for every $(E,\Pi)\in\mathcal{G}(N)$ and for any $i \in N$ such that 
$d_i (S^{-i}|E, \Pi) = p$, the following holds:
\[
\phi_i(E, \Pi) = \sum_{h=1}^p \phi_i(E_h^i, \Pi),
\]
where $\{E_h^i\}_{h=1}^p$ is the family of all possible sets of edges obtained by removing $p - 1$ edges connecting node $i$ with nodes in $N(i, S^{-i}|E, \Pi)$.
\end{axiom}
Linearity axiom reflects the principle that the power of a firm can be decomposed into the contributions derived from individual interactions with other firms. Specifically, if a firm $i$ interacts with $p$ other firms, its total power is the sum of its power across each of these interactions considered separately. This axiom captures the idea that power accumulates additively over independent connections.

\begin{axiom}[Union Indifference (UI)]\label{Ax-UI}
A power index $\phi$ satisfies the \emph{Union Indifference} property if, for every $(E,\Pi) \in \mathcal{G}(N)$, for all $i \in N$, and for any $S \in \Pi$ such that $N(i, S) = \emptyset$, the following holds:
\[
\phi_i(E, \Pi) = \phi_i(E, \tilde{\Pi}),
\]
where $\tilde{\Pi} \in \mathbf{\Pi}_{r-1}(N)$ is obtained from $\Pi$ by merging $S$ with some other union $S' \in \Pi$, with $S' \neq S^i$.
\end{axiom}
Union Indifference axiom states that the power of a firm is sometimes unaffected by the merging of technological sectors. Specifically, if a firm $i$ does not interact with a given technological sector or union, its power remains unchanged even if such a sector merges with another one. This highlights that only active connections influence power, while passive structural changes in unrelated parts of the network do not.

To introduce the final two axioms, we first define a special class of elementary technological networks. Within the model of technological networks, we define simpler structures of graphs with a priori unions, namely \emph{unanimity graphs}.

\begin{definition}[Unanimity Graph]\label{unanumity}
A graph $(E, \Pi) \in \mathcal{G}(N)$ is called a \emph{unanimity graph} if the following conditions hold:
\begin{description}
    \item[i)] $\Pi \in \Pi_2(N)$, i.e., the node set $N$ is partitioned into exactly two disjoint unions;
    \item[ii)] there exists at least a node $i$ such that:
    \begin{description}
    \item[-]  $\epsilon^{\text{ext}} = d_i^{\text{ext}} \geq 1$; 
    \item[-] each node $j \in N(i, S^i)$ satisfies: $d_j^{\text{int}} = 1$;
    \item[-] each node $j \notin N(i, S^i)$ satisfies: $d_j^{\text{int}} = 0$.
    \end{description}
\end{description}
\end{definition}
Figure \eqref{fig:firstUgraph} depicts an example of unanimity graphs. 
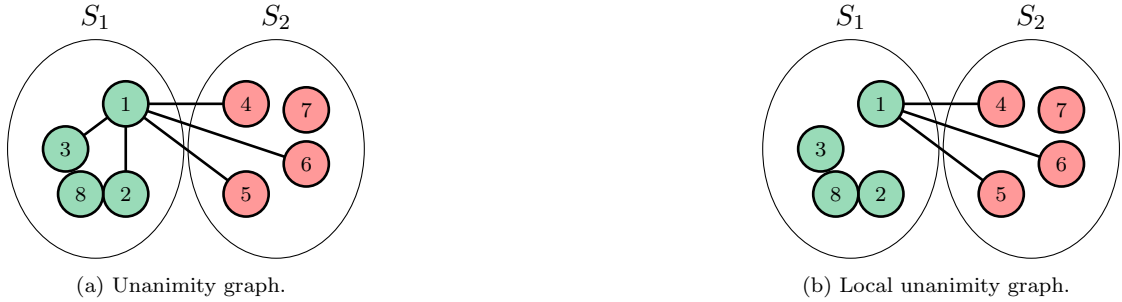
\begin{figure}[htbp]
\centering
\subfloat[Unanimity graph.\label{fig:firstUgraph}]{
 \begin{tikzpicture}[scale=0.40]
\Vertex[size=.6,color=Green,opacity=.4,label=1]{1}
\Vertex[x=0,y=-3.0,size=.6,color=Green,opacity=.4,label=2]{2}
\Vertex[x=-2,y=-1.5,size=.6,color=Green,opacity=.4,label=3]{3}
\Vertex[x=4,y=0.0,size=.6,color=red,opacity=.4,label=4]{4}
\Vertex[x=4,y=-3.0,size=.6,color=red,opacity=.4,label=5]{5}
\Vertex[x=6,y=-2,size=.6,color=red,opacity=.4,label=6]{6}
\Vertex[x=6,y=-0.2,size=.6,color=red,opacity=.4,label=7]{7}
\Vertex[x=-1.5,y=-3,size=.6,color=Green,opacity=.4,label=8]{8}
\Edge[lw=1,color=black,bend=0](1)(4)
\Edge[lw=1,color=black,bend=0](1)(5)
\Edge[lw=1,color=black,bend=0](1)(6)
\Edge[lw=1,color=black,bend=0](1)(3)
\Edge[lw=1,color=black,bend=0](1)(2)
\node[draw, ellipse, fit=(1) (2) (3), label=above:$S_1$] {};
\node[draw, ellipse, fit=(4) (5) (6), label=above:$S_2$] {};
\end{tikzpicture} }\hfill
\subfloat[Local unanimity graph.\label{fig:secondUgraph}]{ \begin{tikzpicture}[scale=0.40]
\Vertex[size=.6,color=Green,opacity=.4,label=1]{1}
\Vertex[x=0,y=-3.0,size=.6,color=Green,opacity=.4,label=2]{2}
\Vertex[x=-2,y=-1.5,size=.6,color=Green,opacity=.4,label=3]{3}
\Vertex[x=4,y=0.0,size=.6,color=red,opacity=.4,label=4]{4}
\Vertex[x=4,y=-3.0,size=.6,color=red,opacity=.4,label=5]{5}
\Vertex[x=6,y=-2,size=.6,color=red,opacity=.4,label=6]{6}
\Vertex[x=6,y=-0.2,size=.6,color=red,opacity=.4,label=7]{7}
\Vertex[x=-1.5,y=-3,size=.6,color=Green,opacity=.4,label=8]{8}
\Edge[lw=1,color=black,bend=0](1)(4)
\Edge[lw=1,color=black,bend=0](1)(5)
\Edge[lw=1,color=black,bend=0](1)(6)
\node[draw, ellipse, fit=(1) (2) (3), label=above:$S_1$] {};
\node[draw, ellipse, fit=(4) (5) (6), label=above:$S_2$] {};
\end{tikzpicture} }
\caption{Examples of two unanimity graphs with a priori unions.}
\end{figure}

\begin{axiom}[Constant Increment (ConI)]\label{Ax-ci}
A power index $\phi$ satisfies the \emph{Constant Increment} axiom 
if the following is true. Let $(E, \Pi), \, (E', \Pi) \in \mathcal{G}(N)$ be two unanimity graphs.
Suppose that node $i \in N$ has the \emph{maximum external degree} 
in both graphs with a priori unions:
\[
d_i^{\text{ext}}(E, \Pi) = \epsilon^{\text{ext}}(E, \Pi) \quad \text{and} \quad d_i^{\text{ext}}(E', \Pi) = \epsilon^{\text{ext}}(E', \Pi).
\]
Then, for any node $k \in S^i$ such that $\{i,k\} \notin E$ and $\{i,k\} \notin E'$, we have:
\begin{multline*} 
\sum_{j \in N(i, S^{-i}| E, \Pi)} \big[ \phi_j(E \cup \{\{i,k\}\}, \Pi) - \phi_j(E, \Pi) \big]
\\ =
\sum_{j \in N(i, S^{-i}| E', \Pi)} \big[ \phi_j(E' \cup \{\{i,k\}\}, \Pi) - \phi_j(E', \Pi) \big].
\end{multline*}
\end{axiom}
Constant Increment axiom models a situation where a firm $i$, which 
is exposed to external technology inflows 
(i.e., it receives many licenses from outside its technological sector), 
decides to share some of these technologies with a new peer firm 
$k$ inside its own union. The axiom states 
that the aggregate impact on the external firms 
— in terms of changes to their power index — does 
not depend on the topology of the technological markets. 
This captures the idea that internal sharing of externally 
sourced technologies has a scalable effect on the rest of the network.
For the next and last axiom we need to restrict the above class of unanimity graphs. 
\begin{definition}[Local Unanimity Graph]
A graph $(E,\Pi)\in\mathcal{G}(N)$ is a \emph{local unanimity graph} if 
it satisfies Definition \ref{unanumity} with the following supplementary property to the condition $ii)$: $d_i^{\text{int}}=0$.   
\end{definition}
Figure \eqref{fig:secondUgraph} depicts an example of a local unanimity graph.  

\begin{axiom}[Balancedness (Ba)]\label{Ax-Ba}
Let $(E,\Pi)\in\mathcal{G}(N)$ be a \emph{local unanimity graph}, and let $i \in N$ be a node such that $d_i^{\text{ext}} = \epsilon^{\text{ext}}$. A power index $\phi$ satisfies \emph{Balancedness} if:
\begin{equation*}
    \phi_i(E, \Pi)  -  \sum_{j \in N(i, S^{-i}|E,\Pi)} \phi_j(E, \Pi) = 0.
\end{equation*}
\end{axiom}
Balancedness axiom expresses a form of power conservation in simple technological networks: the power held by a pivotal firm $i$ is exactly balanced by the total power distributed among its external neighbors. In such minimal networks, the power of firm $i$ is counterbalanced by the collective power it confers to its external neighbors.\\

The next result states that the $(\alpha,\beta)$-power index satisfies the proposed axioms.
\begin{theorem}\label{teo:suff}
Let $(\alpha,\beta)\in\mathbb{R}^2$. The $(\alpha,\beta)$-power index $\varphi^{(\alpha,\beta)}$ satisfies Axioms \ref{Ax-en} (EN), \ref{Ax-ld} (IUN), \ref{Ax-a} (A), \ref{Ax-l} (Li), \ref{Ax-UI} (UI),\ref{Ax-ci} (ConI) and \ref{Ax-Ba} (Ba).
\end{theorem}
The independence of the proposed Axioms is proven as follow.
\begin{theorem}\label{indipendence}
Axioms \ref{Ax-en} (EN), \ref{Ax-ld} (IUN), \ref{Ax-a} (A), \ref{Ax-l} (Li), \ref{Ax-UI} (UI), \ref{Ax-ci} (ConI) and \ref{Ax-Ba} (Ba) are logically independent.   \end{theorem}
The next theorem is the main result of this section and it fully characterizes the $(\alpha,\beta)$-power index.
\begin{theorem}\label{main}
Let $\phi$ be a power index satisfying Axioms~\ref{Ax-en} (EN), \ref{Ax-ld} (IUN), \ref{Ax-a} (A), \ref{Ax-l} (Li), \ref{Ax-UI} (UI), \ref{Ax-ci} (ConI), and \ref{Ax-Ba} (Ba). Then, there exists a vector $(\alpha,\beta) \in \mathbb{R}^2$ such that $\phi = \varphi^{(\alpha,\beta)}$.
\end{theorem}

We conclude the section by observing that, if we restrict to positive power indices and we modify Axiom \ref{Ax-ci} (ConI), we ensure that $\alpha$ and $\beta$ are non negative. 

\begin{axiom}[Edge positivity (EP)]\label{Ax-pen}
A power index $\phi$ satisfies the \emph{Edge positivity} property if, for every partition 
$\Pi \in \mathbf{\Pi}_r(N)$ with $r \leq n$ and $i\in N$, it holds that:
\begin{enumerate}[(i)]
    \item $\phi_i(\emptyset, \Pi) = 0$,
    \item $\phi_i(E,\Pi)\geq 0$, $\quad\forall E\in\mathbf{E}(N)$. 
\end{enumerate}
\end{axiom}
\begin{axiom}[Positive Constant Increment (PConI)]\label{Ax-pci}
A power index $\phi$ satisfies the \emph{Positive Constant Increment} axiom 
if Axiom \ref{Ax-ci} (ConI) is satisfied and, under the same notation, 
\[
\sum_{j \in N(i, S^{-i}| E, \Pi)} \big[ \phi_j(E \cup \{\{i,k\}\}, \Pi) - \phi_j(E, \Pi) \big]
\geq0.
\]
\end{axiom}
Then next theorem characterized the $(\alpha,\beta)$-power index introduced in Section \ref{sect1}. 
\begin{theorem}\label{mainp}
Let $\phi$ be a power index satisfying Axioms~\ref{Ax-pen} (EP), \ref{Ax-ld} (IUN), \ref{Ax-a} (A), \ref{Ax-l} (Li), \ref{Ax-UI} (UI), \ref{Ax-pci} (PConI), and \ref{Ax-Ba} (Ba). Then, there exists a vector $(\alpha,\beta) \in \mathbb{R}^2_{\geq 0}$ such that $\phi = \varphi^{(\alpha,\beta)}$.
\end{theorem}

\section*{Concluding Remarks}
This paper developed a novel family of power indices designed for technological sharing networks, where firms interact  within sectors but do not compete across them. These indices are defined over undirected graphs with a priori unions and balance degree-based and rescaled eigenvector centrality via interpretable market coefficients. This structure reflects the dual nature of these environments: part collaborative, part competitive.
We analyzed structural properties of the indices—such as monotonicity—and highlighted how they capture subtle dependencies in network topology. While they are sensitive to local changes, we also observed a tendency for major actors to remain stable in the absence of strong cross-sectoral spillovers. Rankings induced by these indices offer meaningful insights into technological power concentration, especially in regimes of limited externalities.
A key theoretical contribution of this work lies in the cooperative game-theoretic foundation we provide: we show that the proposed indices coincide with the Shapley values of a family of transferable utility (TU) games defined on graphs with a priori unions. However, due to subadditivity, these games often exhibit an empty core, signaling structural instability in the cooperative arrangements they model.
Crucially, we complement this cooperative approach with a normative axiomatization of the indices. The axioms proposed are shown to be independent and reflect desirable principles. This axiomatic foundation not only clarifies the rationale behind the construction of the indices, but also strengthens their interpretability and applicability in comparative analysis across different network configurations.
Our work opens several promising avenues for future research. These include exploring dynamic settings in which coalition structures evolve endogenously in response to shifts in power distribution; extending the model to weighted networks that capture varying intensities or types of technological relationships; and further developing the axiomatic foundations of power indices on graphs with a priori unions—particularly in contexts where the sum-constant (the sum of indices over nodes is constant across graphs) axiom no longer holds.

\section*{Acknowldgements}
Earlier versions of this work were presented at the Annual Meeting of the Italian Association for Mathematics Applied to Social and Economic Sciences (AMASES), held at the Department of 
Economics, University of Palermo, in 2022, and at the 16th European meeting on game theory (SING19), hosted by the Department of Economics, University of Besançon, in 2024. 
We are indebted to the participants of these conferences, as well as to attendees of related  seminars, for their constructive feedback and insightful suggestions.

Francesco Ciardiello would like to express his sincere gratitude to  L\'{a}szl\'{o} K\'{o}czy 
and Stefano Moretti for their valuable comments and discussions, which contributed 
to the refinement of the ideas presented in this paper.

\section*{Declaration of Competing Interest}
The authors declare that they have no competing interests to disclose. This research was conducted independently and did not receive any specific grant from any commercial organization. Michele Aleandri is a member of the GNAMPA of the Istituto Nazionale di Alta Matematica (INdAM) and is supported by the Project of Significant National Interest – PRIN 2022 titled “Impact of the Human Activities on the Environment and Economic Decision Making in a Heterogeneous Setting: Mathematical Models and Policy Implications” (Cineca Code: 20223PNJ8K; CUP: I53D23004320008). Francesco Ciardiello acknowledges support from the University of Salerno through the Fondo di Ateneo per la Ricerca di Base (FARB) under projects 300399FRB23CIARD and 300399FRB24CIARD, both titled “Game Theory and Applications.”

\section*{Appendix: Proofs}\label{app:proofs1}

\begin{proof}[\bf Proof of Proposition~\ref{prop10}]
	Assume $\beta = \alpha$. By definition of the $(\alpha,\beta)$-power index, we have, $\forall (E,\Pi)\in\mathcal{G}(N)$:
	\begin{align*}
		\varphi_i^{(\alpha,\alpha)}(E,\Pi)
		&=  \alpha \left( d_i(S^{-i}) + \sum_{j \in N(i, S^{-i})} \frac{d_j(S^{-i})}{d_j(S^{i})} \right).
	\end{align*}
	Since $d_i(S^{-i}) = |N(i, S^{-i})|$, we can rewrite the expression as follows:
	\begin{align*}
		\varphi_i^{(\alpha,\alpha)}(E,\Pi)
		&= \alpha \left( \sum_{j \in N(i, S^{-i})} 1 + \sum_{j \in N(i, S^{-i})} \frac{d_j(S^{-i})}{d_j(S^{i})} \right) \\
		&= \alpha \sum_{j \in N(i, S^{-i})} \left( 1 + \frac{d_j(S^{-i})}{d_j(S^{i})} \right) \\
		&= \alpha \sum_{j \in N(i, S^{-i})} \frac{d_j(S^{i}) + d_j(S^{-i})}{d_j(S^{i})} \\
		&= \alpha \sum_{j \in N(i, S^{-i})} \frac{d_j}{d_j(S^{i})}.
	\end{align*}
	This concludes the proof.
\end{proof}

\begin{proof}[\bf Proof of Proposition~\ref{monot}]
	Let $(E,\Pi)\in\mathcal{G}(N)$ with \( i \in N \) be an incomplete bridge node, and let \( \ell \in \{1, \dots, r\} \) be such that
	\[
	\emptyset \subsetneq N(i, S_\ell) \subsetneq S_\ell.
	\]
	Then, there exist two nodes \( j, m \in S_\ell \) such that \( j \in N(i) \) and \( m \notin N(i) \). Consider the extended set of links \( E' := E \cup \{\{i, m\}\} \), and let \( (E', \Pi) \) be the updated graph with a priori unions.
	For the sake of simplicity, we define 
	$$a_i(E):= d_i(S^{-i}|E,\Pi), \quad b_i(E):=\dsum_{h \in N(i, S^{-i}|E,\Pi)} \dfrac{d_h(S^{-i}|E,\Pi)}{d_h(S^{i}|E,\Pi)}$$ 
	for each node $i\in N$. For each node \( h \in N \), let
	\[
	\varphi_h^{(\alpha,\beta)}(E, \Pi) = \alpha a_h(E) + \beta b_h(E).
	\]
	First, observe that the number of neighbors of \( i \) outside its own union increases:
	\[
	a_i(E') = d_i(S^{-i}|E',\Pi) = d_i(S^{-i}|E,\Pi) + 1 = a_i(E) + 1 > a_i(E).
	\]
	Moreover, since \( N(i, S^{-i}|E',\Pi) = N(i, S^{-i}|E,\Pi) \cup \{m\} \), we have:
	\begin{align*}
		b_i(E') = & \sum_{h \in N(i, S^{-i}|E',\Pi)} \frac{d_h(S^{-i}|E',\Pi)}{d_h(S^i|E',\Pi)} \\
		& = \frac{d_m(S^{-i}|E,\Pi)}{d_m(S^i|E,\Pi) + 1} + \sum_{h \in N(i, S^{-i}|E,\Pi)} \frac{d_h(S^{-i}|E,\Pi)}{d_h(S^i|E,\Pi)} 
		= \frac{d_m(S^{-i}|E,\Pi)}{d_m(S^i|E,\Pi) + 1} + b_i(E) \\
		&\geq b_i(E).
	\end{align*}
	
	It follows that \( \varphi_i^{(\alpha,\beta)}(E', \Pi) > \varphi_i^{(\alpha,\beta)}(E, \Pi) \), if $\alpha>0$, thus proving the first part of the proposition.\\
	
	Now consider the case of node \( j \in S_\ell \) with \( j \in N(i) \), and assume \( \beta > 0 \).  In the modified network, the neighborhood of \( j \) outside its own union remains unchanged, hence:
	\[
	a_j(E') = d_j(S^{-j}|E',\Pi) = d_j(S^{-j}|E,\Pi) = a_j(E).
	\]
	For the second component:
	\[
	b_j(E') = \sum_{h \in N(j, S^{-j}|E',\Pi)} \frac{d_h(S^{-j}|E',\Pi)}{d_h(S^j|E',\Pi)} 
	= \frac{d_i(N \setminus S_\ell|E,\Pi)}{d_i(S_\ell|E,\Pi) + 1} 
	+ \sum_{\substack{h \in N(j, S^{-j}|E,\Pi) \\ h \neq i}} \frac{d_h(N \setminus S_\ell|E,\Pi)}{d_h(S_\ell|E,\Pi)}.
	\]
	By the definition of incomplete bridge, we know that \( d_i(N \setminus S_\ell|E,\Pi) > 0 \) we know that
	\[
	0\not=\frac{d_i(N \setminus S_\ell|E,\Pi)}{d_i(S_\ell|E,\Pi) + 1} 
	< \frac{d_i(N \setminus S_\ell|E,\Pi)}{d_i(S_\ell|E,\Pi)},
	\]
	and we conclude that \( b_j(E') < b_j(E) \), and hence
	$
	\varphi_j^{(\alpha,\beta)}(E', \Pi) < \varphi_j^{(\alpha,\beta)}(E, \Pi).
	$
	This completes the proof.
\end{proof}

\begin{proof}[\bf Proof of Proposition~\ref{thm change}]
	For the sake of simplicity, we define 
	$$a_i:= d_i(S^{-i}), \quad b_i:=\dsum_{h \in N(i, S^{-i})} \dfrac{d_h(S^{-i})}{d_h(S^{i})}$$ 
	for each node $i\in N$.
	Let $(E,\Pi)\in\mathcal{G}(N)$ and \( i, j \in N \) be such that \( b_i - b_j \neq 0 \), so that Assumption~\eqref{assumption change} holds for the node pair. We define the critical threshold:
	\[
	\rho^* = \rho^*(i, j) = -\frac{a_i - a_j}{b_i - b_j}>0.
	\]
	By the definition of the power index, we have:
	\[
	\varphi_i^{(\alpha,\beta)} - \varphi_j^{(\alpha,\beta)} = \alpha(a_i - a_j) + \beta(b_i - b_j).
	\]
	Dividing both sides by \( \alpha > 0 \), and letting \( \rho = \frac{\beta}{\alpha} \), we obtain:
	\[
	\varphi_i^{(\alpha,\beta)} - \varphi_j^{(\alpha,\beta)} = (a_i - a_j) + \rho(b_i - b_j).
	\]
	Hence, the sign of \( \varphi_i^{(\alpha,\beta)} - \varphi_j^{(\alpha,\beta)} \) depends on whether \( \rho \) is greater or smaller than \( \rho^* \). More precisely we obtain the following alternative cases:
	\begin{itemize}
		\item[-] If \( b_i - b_j > 0 \) then \(a_i - a_j< 0 \) and we have the following inequalities:
		\[
		\begin{cases}
			\varphi_i^{(\alpha,\beta)} > \varphi_j^{(\alpha,\beta)} & \text{if and only if } \rho > \rho^*, \\
			\varphi_i^{(\alpha,\beta)} < \varphi_j^{(\alpha,\beta)} & \text{if and only if } \rho < \rho^*.
		\end{cases}
		\]
		\item[-] If \( b_i - b_j < 0 \) then \(a_i - a_j > 0 \) and
		we have the following inequalities:
		\[
		\begin{cases}
			\varphi_i^{(\alpha,\beta)} > \varphi_j^{(\alpha,\beta)} & \text{if and only if } \rho < \rho^*, \\
			\varphi_i^{(\alpha,\beta)} < \varphi_j^{(\alpha,\beta)} & \text{if and only if } \rho > \rho^*.
		\end{cases}
		\]
	\end{itemize}
	This concludes the proof.
\end{proof}

\begin{proof}[\bf Proof of Proposition~\ref{prop29}]
	Let $(E,\Pi)\in\mathcal{G}(N)$ a graph that satisfies the restricted spillover property, we aim to show that \( b_i = 0 \) for every \( i \in N \), where
	\[
	b_i = \sum_{h \in N(i, S^{-i})} \frac{d_h(S^{-i})}{d_h(S^i)} .
	\]
	
	By the restricted spillover property, for each node \( i \in N \), there exists a unique union index \( \psi(i) \in \{1, \dots, r\} \) such that all neighbors of \( i \) belong to the union \( S_{\psi(i)} \), with \( S_{\psi(i)} \neq S^i \). In particular, this implies:
	\[
	N(i) \subseteq S_{\psi(i)} \subseteq S^{-i}.
	\]
	
	Let \( h \in N(i, S^{-i}) \). Since \( h \) is a neighbor of \( i \) and lies in \( S_{\psi(i)} \), the restricted spillover property applied to node \( h \) implies that all its neighbors are contained within a unique union \( S_{\psi(h)} \), distinct from its own union \( S^h \). 
	Noting that \( i \in N(h) \subseteq S_{\psi(h)} \) and \( i \in S^i \), we conclude that \( S_{\psi(h)} = S^i \). Therefore, all neighbors of node \( h \) are contained in \( S^i \), which implies that \( h \) has no neighbors in \( S^{-i} \). Consequently we have
	$d_h(S^{-i}) = 0.$ It follows that the contribution \( b_i \) to the power index of node \( i \) is zero.

	Hence, for every \( i \in N \), the power index reduces to \( \varphi_i^{(\alpha,\beta)}(E,\Pi) = \alpha d_i(S^{-i}) \), and the resulting ranking \( \mathcal{R}^0 \) is invariant with respect to the coefficient \( \alpha > 0 \). This concludes the proof.
\end{proof}

\begin{proof}[\bf Proof of Proposition \ref{cic}]
	Let \( (N, \nu; (E, \Pi), \alpha, \beta) \) be the TU-game associated to $(E, \Pi)\in\mathcal{G}(N)$ and $(\alpha, \beta)\in\mathbb{R}_{\geq0}^2$.We prove the first part of the Proposition.
	Let \( T \) and \( S \) be two disjoint coalitions within \( S_\ell \), for some \( \ell \in \{1, \dots, r\} \). By definition, we have:
	\[
	\begin{aligned}\label{ineqa}
		\nu(T \cup S) &= \alpha \sum_{i \in T \cup S} d_i^{\text{ext}} + \beta \sum_{h \in N(T \cup S) \cap (N \setminus S_\ell)} d_h(N\setminus S_\ell)\\
		&= \alpha \sum_{i \in T} d_i^{\text{ext}} + \alpha \sum_{i \in S} d_i^{\text{ext}} 
		+ \beta \sum_{h \in N(T) \cap (N \setminus S_\ell)}d_h(N\setminus S_\ell) 
		+ \beta \sum_{h \in N(S) \cap (N \setminus S_\ell)}d_h(N\setminus S_\ell) \\
		&\quad - \beta \sum_{h \in N(T) \cap N(S) \cap (N \setminus S_\ell)}d_h(N\setminus S_\ell) \\
		&\leq \nu(T) + \nu(S).
	\end{aligned}
	\]
	Now let \( T \) and \( S \) be any two disjoint coalitions in \( N \). The inequality  
	$
	\nu(T \cup S) \leq \nu(T) + \nu(S)
	$
	follows immediately from the linearity of \( \nu \) with respect to the elements of the partition \( \Pi \), and from the case just discussed.
	
	We prove the second part of the Proposition. Assume \( \beta = 0 \). Let \( T \) and \( S \) be two disjoint coalitions within \( S_\ell \), for some \( \ell \in \{1, \dots, r\} \). The additivity of $\nu_1$ is straightforward and implies $\nu(T\cup S)=\nu(T)+\nu(S)$.
	Now, consider two arbitrary disjoint coalitions \( T \) and \( S \) in \( N \). The value function \( \nu \) can be expressed as follows:
	\[
	\begin{aligned}
		\nu(T \cup S) &= \sum_{\ell=1}^r \nu_{1}\left((T \cup S) \cap S_\ell\right) = \sum_{\ell=1}^r \nu_{1}\left(T \cap S_\ell\right) + \sum_{\ell=1}^r \nu_1\left(S \cap S_\ell\right) \\
		&= \nu(T) + \nu(S).
	\end{aligned}
	\]
	The second equality comes from the previous case when coalition are within the same union.
	Therefore, the game is additive if \( \beta = 0 \).

	We prove the third part of the Proposition. 
	Suppose each coalition \( S_\ell \) in the partition \( \Pi \) consists of a single node \( i_\ell \), i.e., \( S_\ell = \{i_\ell\} \) for \( \ell \in \{1, \dots, n\} \). Let \( T \) and \( S \) be two arbitrary disjoint coalitions in \( N \). Then:
	
	\[
	\begin{aligned}
		\nu(T \cup S) &= \sum_{\ell=1}^{n} \nu\left((T \cup S) \cap S_\ell\right) = \sum_{\substack{\ell=1 \\ i_\ell \in T \cup S}}^{n}
		\nu\left(\{i_\ell\}\right) \\[0.3em]
		&= \sum_{\substack{\ell=1 \\ i_\ell \in T}}^{n} \nu\left(\{i_\ell\}\right) +
		\sum_{\substack{\ell=1 \\ i_\ell \in S}}^{n} \nu\left(\{i_\ell\}\right) = \nu(T) + \nu(S).
	\end{aligned}
	\]
	Thus, the game is additive when each union in \( \Pi \) is a singleton.
	
	We prove the fourth sentence of the proposition.
	Let $i,j$ be two nodes having in common a neighbors that lie externally to their own union.
	Let $T=\{i\}$ and $S=\{j\}$. We propose the same equalities in \eqref{ineqa}.
	We note that  the negative term is strictly negative since
	$N(T) \cap N(S) \cap (N \setminus S_\ell)\neq \emptyset$.
	Therefore, we obtain: $\nu(T\cup S)< \nu(T)+\nu(S)$.
	Then, additivity of the game is violated.
	
\end{proof}

\begin{proof}[\bf Proof of Proposition \ref{prop9}]
	Let \( (N, \nu; (E, \Pi), \alpha, \beta) \) be the TU-game associated to $(E, \Pi)\in\mathcal{G}(N)$ and $(\alpha, \beta)\in\mathbb{R}_{\geq0}^2$. Let \( T \subseteq N \) be a nonempty coalition that is not entirely contained in any union \( S_\ell \), with \( \ell \in \{1, \ldots, r\} \).
	The thesis is trivially true if \( T \) contains two elements.
	Assume by induction that \( \Delta_{\nu}(T) = 0 \) for all coalitions \( T \subseteq N \) with \( 2 \leq |T| < m \), which are not contained in any union \( S_\ell \).
	We now prove the claim for \( |T| = m \), under the assumption that \( T \) is not contained in any union \( S_\ell \).
	
	By the definition of Harsanyi dividends, we have:
	\[
	\Delta_{\nu}(T) = \nu(T) - \sum_{R \subsetneq T} \Delta_{\nu}(R).
	\]
	Since \( T \) is not fully contained in any \( S_\ell \), we can write:
	\[
	\nu(T) = \sum_{\ell=1}^{r} \nu(T \cap S_\ell),
	\]
	and using the linearity of the decomposition, we obtain:
	\[
	\sum_{R \subsetneq T} \Delta_{\nu}(R) = \sum_{\ell=1}^{r} \sum_{R \subsetneq T \cap S_\ell} \Delta_{\nu}(R),
	\]
	where the inner sum uses the inductive hypothesis (since \( |R| < m \)) and the fact that each \( T \cap S_\ell \subsetneq T \).
	
	Therefore, we can write:
	\[
	\Delta_{\nu}(T) = \sum_{\ell=1}^{r} \left( \nu(T \cap S_\ell) - \sum_{R \subsetneq T \cap S_\ell} \Delta_{\nu}(R) \right).
	\]
	Observe that for each \( \ell \), the term inside the sum equals \( \Delta_{\nu}(T \cap S_\ell) \). But since \( T \cap S_\ell \subsetneq T \) and \( |T \cap S_\ell| < m \), by the inductive hypothesis we have \( \Delta_{\nu}(T \cap S_\ell) = 0 \).
	
	Hence:
	\[
	\Delta_{\nu}(T) = \sum_{\ell=1}^{r} \Delta_{\nu}(T \cap S_\ell) = 0.
	\]
	This completes the induction and proves the claim.
\end{proof}

To prove Theorem \ref{pot} we premise the following 
key Lemma.
\begin{lemma}\label{lemmae2}
	Let \( (N, \nu; (E, \Pi), \alpha, \beta) \in \mathcal{U}(N) \) be the proposed TU-game, where \( \Pi = \{S_1, \dots, S_r\} \) with \( r \geq 2 \). Then, for every \( T \subseteq N \), the following identity holds:
	\begin{equation}
		\label{lemma201}
		\sum_{i \in T} \left[ P(T) - P(T \setminus \{i\}) \right] = \nu(T).
	\end{equation}
\end{lemma}
\begin{proof}
	If \( T = \emptyset \), then \eqref{lemma201} trivially holds since both sides are zero.\\
	Let \( \ell \in \{1, \dots, r\} \), and let \( T \subseteq S_\ell \) be a non-empty set and write $ P(T) = \alpha A(T) + \beta B(T)$, where $A(T) = \sum_{j \in T} d_j(S^{-j})$ and $B(T) = \sum_{h \in N \setminus S_\ell} \lambda_h(T)$. Fix \( i \in T \). By construction, we have:
	\[
	A(T) - A(T \setminus \{i\}) = d_i(S^{-i}).
	\]
	Now consider any \( h \in N \setminus S_\ell \). Then:
	\[
	\lambda_h(T) - \lambda_h(T \setminus \{i\}) =
	\begin{cases}
		0& \text{if } h \notin N(i), \\
		&\\
		\dfrac{d_h(N \setminus S_\ell)}{d_h(T)}& \text{if } h \in N(i).
	\end{cases}
	\]
	Thus, the variation in \( B \) is given by:
	\[
	B(T) - B(T \setminus \{i\}) = \sum_{h \in N(i) \cap (N \setminus S_\ell)} \frac{d_h(N \setminus S_\ell)}{d_h(T)}.
	\]
	Summing over all \( i \in T \), we obtain:
	\begin{equation}
		\label{lemma202}
		\begin{aligned}
			\sum_{i \in T} \left[ P(T) - P(T \setminus \{i\}) \right] &= 
			\alpha \sum_{i \in T} \left[ A(T) - A(T \setminus \{i\}) \right]
			+ \beta \sum_{i \in T} \left[ B(T) - B(T \setminus \{i\}) \right] \\
			&= \alpha \sum_{i \in T} d_i(S^{-i}) + 
			\beta \sum_{i \in T} \sum_{h \in N(i) \cap (N \setminus S_\ell)} \frac{d_h(N \setminus S_\ell)}{d_h(T)} \\
			&= \alpha \sum_{i \in T} d_i^{\text{ext}} +
			\beta \sum_{h \in N(T) \cap (N \setminus S_\ell)} d_h(N \setminus S_\ell) \\
			&= \nu(T).
		\end{aligned}
	\end{equation}
	Now, consider a general coalition \( T \subseteq N \), with \( T \neq \emptyset \). From the additivity property of $P$ and  \(\nu \) over disjoint unions, we have:
	\begin{equation}
		\label{lemma203}
		\begin{aligned}
			\sum_{i \in T} \left[ P(T) - P(T \setminus \{i\}) \right]
			&= \sum_{i \in T} \sum_{\ell=1}^{r} \left[ P(T \cap S_\ell) - P(T \setminus \{i\}) \cap S_\ell) \right] \\
			&= \sum_{\ell = 1}^{r} \nu(T \cap S_\ell) \\
			&= \nu(T),
		\end{aligned}
	\end{equation}
	which concludes the proof.
\end{proof}

\begin{proof}[\bf Proof of Proposition \ref{pot}]
	The proof is an immediate consequence of  Theorem \ref{mainre}  
	and  Lemma \ref{lemmae2}. 
\end{proof}

\begin{proof}[\bf Proof of Theorem~\ref{thmshapley}]
	Let \( (N, \nu; (E, \Pi), \alpha, \beta) \in \mathcal{U}(N) \) be the TU-game associated to $(E, \Pi)\in\mathcal{G}(N)$ and $(\alpha, \beta)\in\mathbb{R}_{\geq0}^2$ and take \( i \in N \). By Theorem~\ref{mainre} and Proposition~\ref{pot}, the Shapley value of player \( i \) can be computed as the marginal contribution of the potential function:
	\[
	Sh_i(N, \nu) = \mathcal{P}(N, \nu) - \mathcal{P}(N \setminus \{i\}, \nu).
	\]
	Using the fact that \( \mathcal{P}(N, \nu) = P(N) \) and \( \mathcal{P}(N \setminus \{i\}, \nu) = P(N \setminus \{i\}) \), we obtain:
	\[
	Sh_i(N, \nu) = P(N) - P(N \setminus \{i\}) = P(S^i) - P(S^i \setminus \{i\}).
	\]
	By the explicit definition of the potential function \( P \), we then have:
	\[
	Sh_i(N, \nu) = \alpha \, d_i(S^{-i}) + \beta \sum_{j \in N(i, S^{-i})} \frac{d_j(S^{-i})}{d_j(S^i)},
	\]
	which coincides with \( \varphi_i^{(\alpha,\beta)}(E,\Pi)\). 
\end{proof}

\begin{proof}[\bf Proof of Proposition \ref{totalpower}.]
	Let \( (N, \nu; (E, \Pi), \alpha, \beta) \in \mathcal{U}(N) \) be the TU-game associated to $(E, \Pi)\in\mathcal{G}(N)$ and $(\alpha, \beta)\in\mathbb{R}_{\geq0}^2$. By the efficiency property of the Shapley value, the total value allocated to all players must equal the value of the grand coalition:
	\begin{align*}
		\sum_{i=1}^n \varphi_i^{(\alpha,\beta)}(E,\Pi)
		&= \nu(N) \\
		&= \sum_{\ell=1}^r \left[ \nu_1(S_\ell) + \nu_2(S_\ell) \right] \\
		&= \alpha \sum_{\ell=1}^r \sum_{j \in S_\ell} d_j^{\mathrm{ext}} 
		\; + \; \beta \sum_{\ell=1}^r \sum_{\substack{h \in N \setminus S_\ell \\ \ell \in u(h)}} d_h(N \setminus S_\ell) \\
		&= 2\alpha\, \epsilon^{\mathrm{ext}} 
		\; + \; \beta \sum_{\ell=1}^r \sum_{\substack{h \in N \setminus S_\ell \\ \ell \in u(h)}} d_h(N \setminus S_\ell).
	\end{align*}
\end{proof}


\begin{proof}[\bf Proof of Theorem \ref{teo:suff}]
	
	Axiom \ref{Ax-en} (EN) is satisfied by definition.
	
	Let $(E,\Pi)\in\mathcal{G}(N)$, with $\Pi\in\mathbf{\Pi}_r(N)$ and $r\leq n$, and let $i\in N$. Define the set $E_i^*$ as in Axiom \ref{Ax-ld} (IUN). Note that the number of neighbors of $i$ outside of $S^i$ remains unchanged in both $(E,\Pi)$ and $(E_i^*,\Pi)$, hence 
	$d_i(S^i)$ is constant. Furthermore, for every $j\in N(i, S^{-i}|E,\Pi)$, the set $E_i^*$ includes all edges $\{j,h\}$ where $h \in N(j|E,\Pi)$, so both $d_j(S^{-i})$ and $d_j(S^i)$ are preserved. Therefore, Axiom \ref{Ax-ld} (IUN) holds.

	Axiom \ref{Ax-a} (A) follows directly, since any permutation $\sigma:N\to N$ does not alter the number of connections of each node, and the function $d_{\cdot}(\cdot)$ is invariant under renaming of nodes.

	Now assume $N(i, S^{-i}) = \{i_1,\ldots,i_p\}$ for some $p\geq 1$. Define the family $\{(E_h^i, \Pi)\}_{h=1}^p$ as in Axiom \ref{Ax-l} (Li). By construction,
	\[
	\varphi^{(\alpha,\beta)}_i(E,\Pi) = \sum_{h=1}^p \left(\alpha + \beta \frac{d_{i_h}( S^{-i})}{d_{i_h}( S^i)}\right) = \sum_{h=1}^p \varphi^{(\alpha,\beta)}_i(E_h^i,\Pi),
	\]
	thus Axiom \ref{Ax-l} (Li) is satisfied.
	
	Consider $i\in N$ and $(E,\Pi)\in\mathcal{G}(N)$ such that $E \neq \emptyset$, $\Pi = (S_1,\ldots,S_r) \in \bm{\Pi}_r$, with $3 \leq r \leq n$. Without loss of generality, suppose $i\in S_1$ and $N(i, S_r) = \emptyset$. Let $\tilde{\Pi} = (S_1, \ldots, S_{r-1} \cup S_r)$. Since
	\[
	N(i, S^{-i}|E, \Pi) = N(i, S^{-i}|E, \tilde{\Pi}),
	\]
	we have $d_{i}(S^{-i}|E,\Pi) = d_{i}(S^{-i}|E,\tilde{\Pi})$. Moreover, for all $j \in N(i, S^{-i}|E,\Pi)$, the values $d_j (S^i)$ and $d_j( S^{-i})$ remain unchanged, hence $\varphi_i^{(\alpha,\beta)}(E,\Pi) = \varphi_i^{(\alpha,\beta)}(E,\tilde{\Pi})$. Therefore, Axiom \ref{Ax-UI} (UI) is satisfied.
	
	Now consider a unanimity graph $(E,\Pi)$ and nodes $i,k \in N$ satisfying the conditions of Axiom \ref{Ax-ci} (ConI). For each $j \in N(i, S^{-i})$, we compute:
	
	\begin{align*}
		\varphi_j^{(\alpha,\beta)}(E \cup \{i,k\}, \Pi) 
		&= \alpha \, d_j (S^{-j}|E \cup \{i,k\}, \Pi) \notag \\
		&\quad 
		+ \beta \sum_{h \in N(j, S^{-j}|E \cup \{i,k\}, \Pi)} 
		\frac{d_h(S^{-j}|E \cup \{i,k\}, \Pi)}{d_h (S^j|E \cup \{i,k\}, \Pi)} \notag \\
		&= \alpha \, d_j (S^{-j}|E, \Pi) \notag \\
		&\quad 
		+ \beta \sum_{h \in N(j, S^{-j}|E, \Pi)} 
		\frac{d_h(S^{-j}|E, \Pi) + \mathbf{1}_{h = i}}{d_h (S^j|E, \Pi)} \notag \\
		&= \varphi_j^{(\alpha,\beta)}(E, \Pi) 
		+ \beta \frac{1}{d_i (S^j|E, \Pi)}.
	\end{align*}

	Summing over all $j \in N(i, S^{-i}|E, \Pi)$, we obtain:
	\begin{align*}
		\sum_{j \in N(i, S^{-i}|E \cup \{i,k\}, \Pi)} \varphi_j^{(\alpha,\beta)}(E \cup \{i,k\}, \Pi) &= \sum_{j \in N(i, S^{-i}|E, \Pi)} \varphi_j^{(\alpha,\beta)}(E, \Pi) \\
		& + \sum_{j \in N(i, S^{-i}|E, \Pi)} \beta \frac{1}{d_i (S^j|E, \Pi)} \\
		&= \sum_{j \in N(i, S^{-i}|E, \Pi)} \varphi_j^{(\alpha,\beta)}(E, \Pi) + \beta.
	\end{align*}
	This implies that the difference is   constant for any $(E,\Pi)$ and then Axiom \ref{Ax-ci} (ConI) holds.
	
	Finally, consider a local unanimity graph $(E, \Pi)$ and $i\in N$ as in Axiom \ref{Ax-Ba} (Ba). By construction, $\varphi_i^{(\alpha,\beta)}(E, \Pi) = \alpha 
	d_i (S^{-i})$, since the neighbors of $i$ in $S^{-i}$ have no further connections. Likewise, for any $j \in N(i, S^{-i})$, we have $\varphi_j^{(\alpha,\beta)}(E,\Pi) = \alpha$, since $i$ has no neighbors in $S^i$. These conditions directly imply the validity of Axiom \ref{Ax-Ba} (Ba).
\end{proof}

\begin{proof}[\bf Proof of Theorem \ref{indipendence}]
	We provide specific power indices that violate the axiom under consideration while still satisfying all the remaining axioms. The list of each axiom under consideration follows:    
	\noindent\vspace{.5\baselineskip} \\
	\textbf{(EN):}  
	Consider the power index $\phi$ defined, $\forall i\in N$ and $\forall(E,\Pi)\in\mathcal{G}(N)$, by 
	\begin{equation*}
		\phi_i(E,\Pi)=\begin{cases}
			1, & \text{ if } E=\emptyset;\\
			\phi_{i}^{(\alpha,\beta)}(E,\Pi), & \text{ otherwise.}
		\end{cases}
	\end{equation*}
	This power index does not satisfy Axiom \ref{Ax-en} (EN) and  satisfies the other axioms in virtue of Theorem \ref{teo:suff}.
	
	\noindent\vspace{.5\baselineskip}\\
	\textbf{(A):} 
	Set $\mathcal{G}^A(N)\subset \mathcal{G}(N)$ the set of graphs with a priori union $(E,\Pi)$ such that $ \{1,2\},\{1,3\}\in E$, $S^2=S^3\neq S^1$. Define the power index $\phi$ in the following way:
	\begin{eqnarray}
		\phi_i(E,\Pi) &=& d_i^{\text{ext}}(E,\Pi), \quad \text{for each } i \in N \text{ if } (E,\Pi) \in \mathcal{G}(N)\setminus \mathcal{G}^A(N); \\
		\phi_i(E,\Pi) &=& 
		\begin{cases}
			d_2^{\text{ext}}(E,\Pi) + \frac{1}{2}, & \text{if } i=2, \\
			d_3^{\text{ext}}(E,\Pi) - \frac{1}{2}, & \text{if } i=3, \\
			d_i^{\text{ext}}(E,\Pi), & \text{otherwise,}
		\end{cases}
		\quad \text{if } (E,\Pi) \in \mathcal{G}^A(N).
	\end{eqnarray}

	The solution $\phi$ does not satisfy Axiom \ref{Ax-a} (A). Consider $(\tilde{E},\tilde{\Pi})$ where $\tilde{E} = \big\{\{1,2\}, \{1,3\}\big\}$, and the coalition structure satisfies the following:  $S^2 = S^3 \neq S^1$. Let $\sigma$ be a permutation such that $\sigma(2) = 3$, $\sigma(3) = 2$, and $\sigma(i) = i$ for all $i \neq 2,3$. After simple computations, we obtain:
	\[
	\phi_2(\tilde{E},\tilde{\Pi}) = \frac{3}{2} \quad \text{while} \quad \phi_{\sigma(2)}(\sigma \tilde{E}, \sigma \tilde{\Pi}) = \phi_3(\tilde{E},\tilde{\Pi}) = \frac{1}{2}.
	\]
	This discrepancy violates the anonymity requirement imposed by Axiom \ref{Ax-a}.
	
	By construction, the solution $\phi$ satisfies Axioms \ref{Ax-en} (EN), \ref{Ax-ld} (IUN), \ref{Ax-UI} (UI), and \ref{Ax-ci} (ConI), since $d_i^{\text{ext}}(E,\Pi) = d_i^{\text{ext}}(E_i^*,\Pi)$ and $d_i^{\text{ext}}(E,\Pi) = d_i^{\text{ext}}(E \cup \{i,k\},\Pi)$. It is straightforward to verify that $d_i^{\text{ext}}$ satisfies Axiom \ref{Ax-l} (Li) for all $i \in N$. 
	
	Observe that if $(E,\Pi) \notin \mathcal{G}^A(N)$, then $(E_h^i,\Pi) \notin \mathcal{G}^A(N)$ for every $i \in N$, where $E_h^i$ is defined as in Axiom \ref{Ax-l} (Li). If instead $(E,\Pi) \in \mathcal{G}^A(N)$, then $(E_h^i,\Pi) \notin \mathcal{G}^A(N)$ for all $i \neq 2,3$, while for $i = 2,3$, the family $\{(E_h^i,\Pi)\}_h$ defined in Axiom \ref{Ax-l} (Li) contains exactly one element in $\mathcal{G}^A(N)$. In this case, we have
	\[
	\phi_i(E,\Pi) = \sum_{h=1}^p \phi_i(E_h^i,\Pi) = d_i^{\text{ext}}(E,\Pi) + \frac{1}{2}(-2i + 5),
	\]
	which shows that Axiom \ref{Ax-l} (Li) is satisfied.

	Finally, consider Axiom \ref{Ax-Ba} (Ba). If a local unanimity graph is not in $\mathcal{G}^A(N)$, then $d_i^{\text{ext}}$ satisfies the axiom for all $i \in N$. If a local unanimity graph is in $\mathcal{G}^A(N)$, then:
	\[
	\sum_{j \in N(1,S^{-1})} \left[\phi_j(E,\Pi) - \phi_1(E,\Pi)\right] = \sum_{j \in N(1,S^{-1})} \left[d_j^{\text{ext}} - d_1^{\text{ext}}\right] = 0,
	\]
	thus confirming that Axiom \ref{Ax-Ba} (Ba) is satisfied in both cases.

	\noindent\vspace{.5\baselineskip}\\
	\textbf{(IUN):} 
	Define, for every $i \in N$ and every $(E, \Pi) \in \mathcal{G}(N)$, the power index $\phi$ as:
	\begin{equation*}
		\phi_i(E,\Pi) = \sum_{h \in N(i, S^{-i})} \sum_{j \in N(h, S^{-i})} d_j,
	\end{equation*}
	with the convention that the value of the sum is zero whenever $N(i, S^{-i}) = \emptyset$ or $N(h, S^{-i}) = \emptyset$. 
	
	Consider the set of agents $N = \{1, \ldots, 4\}$, the set of edges $E = \big\{\{1,2\}, \{2,3\}, \{3,4\}\big\}$, and the partition $\Pi = \big\{\{1\}, \{2,4\}, \{3\}\big\}$. Let $E_1^* = \big\{\{1,2\}, \{2,3\}\big\}$ be the set defined as in Axiom \ref{Ax-ld} (IUN). In this case, we obtain $\phi_1(E, \Pi) = 2$ and $\phi_1(E_1^*, \Pi) = 1$. Therefore, Inter-Union Neighborhood property is violated.
	
	It is straightforward to verify that Axiom \ref{Ax-en} (EN) is satisfied. Consider any permutation $\sigma: N \to N$. Since $d_i = d_{\sigma(i)}$ and $|N(h, S^{-\sigma(i)})| = |N(h, S^{-i})|$ for all $i, h \in N$, it follows that $\phi$ satisfies Axiom \ref{Ax-a} (A).
	
	Next, take any $i \in N$, and let $\{(E_h^i, \Pi)\}_{h=1}^p$ be the family defined as in Axiom \ref{Ax-l} (Li). It holds that $\phi_i(E_h^i, \Pi) = \sum_{j \in N(i_h, S^{-i})} d_j$, hence Axiom \ref{Ax-l} (Li) (Linearity) is satisfied.
	
	Axiom \ref{Ax-UI} (UI) is satisfied by construction, as the index $\phi$ is independent of the number of unions in the partition.
	
	Finally, observe that the unanimity graphs involved in the definitions of Axioms \ref{Ax-ci} (ConI) and \ref{Ax-Ba} (Ba)  consist at most of paths with two edges. In such cases, for every node $i$, $\phi_i(E, \Pi) = 0$. Thus, both axioms are trivially satisfied.

	\noindent\vspace{.5\baselineskip}\\
	\textbf{(Li):} 
	Define the power index $\phi$ as follows:
	\begin{equation*}
		\phi_i(E, \Pi) = 
		\begin{cases}
			\left( \dfrac{\sqrt{d_i (S^{-i})}}{\sum_{j \in N(i, S^{-i})} d_j (S^{i})} \right)^2, & \text{if } N(i, S^{-i}) \neq \emptyset, \\
			0, & \text{otherwise}.
		\end{cases}
	\end{equation*}
	
	Let $N = \{1, 2, 3\}$, and consider the unanimity coalition structure defined by $E = \big\{ \{1,2\}, \{1,3\} \big\}$ and $\Pi = \big\{ \{1\}, \{2,3\} \big\}$. Define $E_1^1 = \big\{ \{1,2\} \big\}$ and $E_2^1 = \big\{ \{1,3\} \big\}$, as in Axiom \ref{Ax-l} (Li). We compute:
	\[
	\phi_1(E, \Pi) = \frac{1}{2}, \quad \text{while} \quad \phi_1(E_1^1, \Pi) + \phi_1(E_2^1, \Pi) = 1 + 1 = 2.
	\]
	Hence, Axiom \ref{Ax-l} (Li) is violated.
	
	It is straightforward to verify that Axioms \ref{Ax-en} (EN), \ref{Ax-ld} (IUN), \ref{Ax-a} (A), \ref{Ax-UI} (UI), and \ref{Ax-ci} (ConI) are satisfied by construction.
	
	Now consider a local unanimity graph $(E, \Pi)$, and without loss of generality, suppose that the external degree of node 1 is $d_1^{\text{ext}} = \epsilon^{\text{ext}} = m\geq 1$. A simple calculation yields:
	\[
	\phi_1(E, \Pi) = \frac{1}{m}, \quad \text{and} \quad \phi_j(E, \Pi) = \frac{1}{m^2}, \quad \forall j \in S^{-1}.
	\]
	Thus, the total contribution of the neighbors of agent 1 outside its union is:
	\[
	\sum_{j \in N(1, S^{-1})} \phi_j(E, \Pi) = \frac{1}{m},
	\]
	which implies that Axiom \ref{Ax-Ba} (Ba) is satisfied.

	\noindent\vspace{.5\baselineskip}\\
	\textbf{(UI):} 
	Define the power index $\phi$ as follows:
	\begin{equation*}
		\phi_i(E, \Pi) = |\Pi|\,d_i^{\text{ext}}.
	\end{equation*}
	
	By construction, $\phi$ does not satisfy Axiom \ref{Ax-UI} (UI), as it explicitly depends on the number of unions. However, it satisfies Axioms \ref{Ax-en} (EN), \ref{Ax-ld} (IUN), \ref{Ax-a} (A), \ref{Ax-l} (Li), \ref{Ax-ci} (ConI), and \ref{Ax-Ba} (Ba).

	\noindent\vspace{.5\baselineskip}\\
	\textbf{(ConI)}: 
	Define the power index $\phi$, for every $i \in N$ and every $(E, \Pi) \in \mathcal{G}(N)$, as follows:
	\begin{equation*}
		\phi_i(E, \Pi) = \sum_{j \in N(i, S^{-i})} d_j.
	\end{equation*}
	
	By construction, $\phi$ satisfies Axioms \ref{Ax-en} (EN),
	\ref{Ax-ld} (IUN),
	\ref{Ax-a} (A), \ref{Ax-l} (Li),  \ref{Ax-UI} (UI), and \ref{Ax-Ba} (Ba).
	
	Let $(E, \Pi)$ be a unanimity graph such that $\epsilon^{\text{ext}} = d_i^{\text{ext}}$. Then we have:
	\begin{align*}
		\sum_{j \in N(i, S^{-i}|E \cup \{i,k\}, \Pi)} \phi_j (E \cup \{i,k\}, \Pi)
		&= \sum_{j \in N(i, S^{-i}|E \cup \{i,k\}, \Pi)} d_j (S^{-i}|E \cup \{i,k\}, \Pi) \\
		&= \sum_{j \in N(i, S^{-i}|E, \Pi)} d_j (S^{-i}|E, \Pi) + \sum_{j \in N(i, S^{-i}|E, \Pi)} 1 \\
		&= \sum_{j \in N(i, S^{-i}|E, \Pi)} \phi_j (E, \Pi) + d_i (S^{-i}|E).
	\end{align*}
	
	However, the term $d_i (S^{-i}|E)$ depends on the choice of $E$ 
	defined in Axiom \ref{Ax-ci} (ConI). As a result, Axiom \ref{Ax-ci} (ConI) is not satisfied.

	\noindent\vspace{.5\baselineskip}\\
	\textbf{(Ba)}:  
	Let $\mathcal{G}^A(N) \subset \mathcal{G}(N)$ denote the subset of graphs with a priori union $(E,\Pi)$ such that $\{1,2\}, \{1,3\} \in E$ and $S^2 = S^3 \neq S^1$. Define the power index $\phi$ as follows:
	\begin{equation*}
		\phi_i(E,\Pi) = 
		\begin{cases}
			d_i^{\text{ext}}(E,\Pi), & \text{if } (E,\Pi) \in \mathcal{G}(N) \setminus \mathcal{G}^A(N), \\
			d_i^{\text{ext}}(E,\Pi) + 1, & \text{if } (E,\Pi) \in \mathcal{G}^A(N) \text{ and } i \in \{2,3\}, \\
			d_i^{\text{ext}}(E,\Pi), & \text{if } (E,\Pi) \in \mathcal{G}^A(N) \text{ and } i \notin \{2,3\}.
		\end{cases}
	\end{equation*}
	
	The power index $\phi$ does not satisfy Axiom \ref{Ax-Ba} (Ba). Consider the graph $(\tilde{E}, \tilde{\Pi})$ with $\tilde{E} = \{\{1,2\}, \{1,3\}\}$ and $S^2 = S^3 \neq S^1$. In this case, we have:
	\[
	\phi_2(\tilde{E}, \tilde{\Pi}) = \phi_3(\tilde{E}, \tilde{\Pi}) = 2, \quad \text{but} \quad \phi_1(\tilde{E}, \tilde{\Pi}) = 2.
	\]
	
	The same considerations discussed in \textbf{(A)} apply here as well. By construction, the index $\phi$ satisfies Axioms \ref{Ax-en} (EN), \ref{Ax-ld} (IUN), \ref{Ax-a} (A), \ref{Ax-l} (Li), \ref{Ax-UI} (UI), and \ref{Ax-ci} (ConI).
\end{proof}

Before stating the proof of the main result, we present three useful lemmas.
\begin{lemma}\label{lemma:IUN+EN}
	Let $\phi$ be a power index that satisfies Axioms~\ref{Ax-en} (EN) and~\ref{Ax-ld} (IUN).  
	Then, for any $(E,\Pi) \in \mathcal{G}(N)$ and any node $i \in N$ with zero external degree, i.e., $d_i^{\text{ext}} = 0$, it holds that
	\begin{equation*}
		\phi_i(E, \Pi) = 0.
	\end{equation*}
\end{lemma}
\begin{proof}
	Take a power index $\phi$ satisfying Axioms \ref{Ax-ld} (IUN) and \ref{Ax-en} (EN) and $(E,\Pi)\in\mathcal{G}(N)$. For every node $i$ that has not external links we have $\phi_i(E,\Pi)=\phi_h(E_i^*,\Pi)=\phi_i(\emptyset,\Pi)=0$, where the first equality follows from  Axiom \ref{Ax-ld} (IUN) and the last one from Axiom \ref{Ax-en} (EN).
\end{proof}

\begin{lemma}\label{lemma:IUN+A}
	Let $\phi$ be a power index satisfying Axioms~\ref{Ax-ld} (IUN) and~\ref{Ax-a} (A).  
	Then, for any $(E,\Pi)\in\mathcal{G}(N)$ such that $\epsilon^{\mathrm{ext}} = d_i(S^{-i}) = p$ for some $i\in N$, it holds that
	\begin{equation*}
		\phi_j(E,\Pi) = \phi_{j'}(E,\Pi),
	\end{equation*}
	for all $j, j' \in N(i,S^{-i})$ such that $S^{j} = S^{j'}$.
\end{lemma}
\begin{proof}
	Take $(E,\Pi)\in\mathcal{G}(N)$ and let $i\in N$ be such that the number of external neighbors is $\epsilon^{\mathrm{ext}} = d_i(S^{-i}) = p$. 
	Let $j, j' \in N(i, S_h)$, where $S_h \in \Pi$ is an a priori union different from $S^i$.
	Consider a permutation $\sigma : N \to N$ such that $\sigma(j) = j'$, $\sigma(j') = j$, and $\sigma(k) = k$ for all $k \neq j, j'$. 
	By Axiom~\ref{Ax-a} (A), we have:
	\[
	\phi_{j}(E,\Pi) = \phi_{\sigma(j)}(\sigma E, \sigma \Pi) = \phi_{j'}(\sigma E, \Pi),
	\]
	since $j$ and $j'$ belong to the same a priori union and, hence, their roles are symmetric.
	Now, apply Axiom~\ref{Ax-ld} (IUN): 
	\[
	\phi_{j'}(\sigma E, \Pi) = \phi_{j'}((\sigma E)^*_{j'}, \Pi) = \phi_{j'}(E_{j'}^*, \Pi) = \phi_{j'}(E, \Pi),
	\]
	where the second equality follows by construction, as $(\sigma E)^*_{j'} = E_{j'}^*$.
	Repeating the same argument for all $p$ nodes in $N(i, S^{-i})$, we obtain the desired result.
\end{proof}

\begin{lemma}\label{lemma:ConI+A+UI}
	Let $\phi$ be a power index that satisfies Axioms~\ref{Ax-a} (A),~\ref{Ax-UI} (UI), and~\ref{Ax-ci} (ConI).  
	Then there exists a constant $\beta \in \mathbb{R}$ such that, for every unanimity graph $(E, \Pi) \in \mathcal{G}(N)$ satisfying the assumptions of Axiom~\ref{Ax-ci} (ConI), the following holds:
	\begin{equation*}
		\sum_{j \in N(i, S^{-i}|E, \Pi)} \left[ \phi_j(E \cup \{\{i,k\}\}, \Pi) - \phi_j(E, \Pi) \right] = \beta.
	\end{equation*}
\end{lemma}
\begin{proof}
	Take a unanimity graph $(E,\Pi)\in\mathcal{G}_N$ as in Axiom \ref{Ax-ci} (ConI) and define 
	\begin{equation*}
		\beta_i^k(E,\Pi)=  \sum_{j \in N(i, S^{-i}|E, \Pi)} \phi_j(E \cup \{i,k\}, \Pi) -  \phi_j(E, \Pi).
	\end{equation*}
	If $n=3$ the lemma follows straightforward because up to a permutation there is a unique unanimity graph that satisfies the hypothesis of Axiom \ref{Ax-ci} (ConI). From now onward, we take $n\geq4$.   \\
	
	We start proving that $\beta_i^k(E,\Pi)=\beta_i(E,\Pi)$, namely it does not depend on the label of chosen node $k$. 
	Without loss of generality we can take $i=1\in S_1$.  Consider the following unanimity graph $E=\big\{\{1,4\}\big\}$ and $\Pi=\big\{\{1,2,3\},N\setminus\{1,2\}\big\}$. Take $\phi$ a power index satisfying Axioms \ref{Ax-a} (A), \ref{Ax-UI} (UI), and \ref{Ax-ci} (ConI) and suppose, by contradiction, that $\beta_1^2(E,\Pi)\neq \beta_1^3(E,\Pi)$. Take the permutation $\sigma:N\to N$ such that $\sigma(2)=3$, $\sigma(3)=2$, and $\sigma(j)=j$, for all $j\neq 2,3$, then 
	\begin{align*}
		\beta_1^2(E,\Pi)&=\phi_4(E\cup\{1,2\},\Pi)-\phi_4(E,\Pi)=\phi_{\sigma(4)}(\sigma(E\cup\{1,2\}),\sigma(\Pi))-\phi_{\sigma(4)}(\sigma(E),\sigma(\Pi))\\
		&= \phi_4(E\cup\{1,3\},\Pi)-\phi_4(E,\Pi)=\beta_1^3(E,\Pi),
	\end{align*}
	where the second equality follows from Axiom \ref{Ax-a} (A). 
	
	Define the unanimity graph $(E^*,\Pi)$ as in Axiom \ref{Ax-ci} (ConI) such that $i\in S_1$ and $E^*=\{\{i,j\}\}$ with $j\in S_2$, then by Axiom \ref{Ax-ci} (ConI) $\beta_i(E,\Pi)=\beta_i(E^*,\Pi)$. Using Axiom \ref{Ax-ci} (UI) we have $\beta_i(E^*,\Pi)=\beta_i(E^*,\Pi^*)$, where $\Pi^*=(S_1^*,S_2^*)$ with   $S_1^*=S_1\cup (S_2\setminus\{j\})$ and $S_2^*= S_2\setminus (S_2\setminus\{j\})=\{j\}$. Taking any permutation $\sigma$ that is $i$-indifferent, namely $\sigma(i)=i$, by Axiom \ref{Ax-a} (A) we have that $\beta_i(E^*,\Pi^*)=\beta_i(\sigma(E^*),\sigma(\Pi^*))$. By the generality of $\sigma$, $\beta_i(E^*,\Pi^*)$ does not depend on the labels of $j$ and it is constant among all the possible pair $(E^*,\Pi^*)$. Then collecting all the arguments there exist a constant $\beta_i^*$ such that $\beta_i^k(E,\Pi)=\beta_i^*$. We conclude the proof showing that $\beta_i^*=\beta$, namely it does not depend on the index $i$. Suppose, by contradiction, that there are $i,j\in N$ such that $\beta_i^*\neq \beta_j^*$; without loss of generality we can take $i=1$ and $j=2$. Consider the following unanimity graph $E=\big\{\{1,3\}\big\}$ and $\Pi=\big\{\{1,2\},\{3\}\big\}$. Take $\phi$ a power index satisfying Axioms \ref{Ax-ci} (ConI) and \ref{Ax-a} (A), then $\phi_3(E\cup\{1,2\},\Pi)-\phi_3(E,\Pi)=\beta_1$. Take the permutation $\sigma:N\to N$ such that $\sigma(1)=2$, $\sigma(2)=1$, and $\sigma(3)=3$, then 
	\begin{align*}
		\beta_1^*&=\phi_3(E\cup\{1,2\},\Pi)-\phi_3(E,\Pi)=\phi_{\sigma(3)}(\sigma(E\cup\{1,2\}),\sigma(\Pi))-\phi_{\sigma(3)}(\sigma(E),\sigma(\Pi))\\
		&= \phi_3(\{2,3\}\cup\{1,2\},\Pi)-\phi_3(\{2,3\},\Pi)=\beta_2^*,
	\end{align*}
	where the second equality follows from Axiom \ref{Ax-a} (A) and the third one from Axiom \ref{Ax-ci} (ConI). This contradicts the hypothesis that $\beta_1\neq\beta_2$.
\end{proof}

\begin{proof}[\bf Proof of Theorem \ref{main}]
	
	Let $\phi$ be a power index satisfying Axioms~\ref{Ax-en} (EN), \ref{Ax-ld} (IUN), \ref{Ax-a} (A), \ref{Ax-l} (Li), \ref{Ax-UI} (UI), and \ref{Ax-ci} (ConI)
	and \ref{Ax-Ba} (Ba).
	The proof is divided into several steps. 
	
	At each step, we examine specific coalition structures with a priori unions, 
	and for each of these, we calculate the power index of all the involved nodes. 
	The objective is to show that the computed power index coincides with the power index \( \varphi^{(\alpha,\beta)} \), where the pair \( (\alpha,\beta) \) will be defined uniquely throughout the steps.
	
	\paragraph{STEP 1} 
	Suppose $r = 2$, and let $(E,\Pi) \in \mathcal{G}(N)$ be such that the external degree $\epsilon^{\mathrm{ext}} = 1$, with $\{i,j\} \in E$, $i \in S_1$, $j \in S_2$, and $N(i,S_1)=N(j,S_2)=\emptyset$ (see for example the figure below).
	
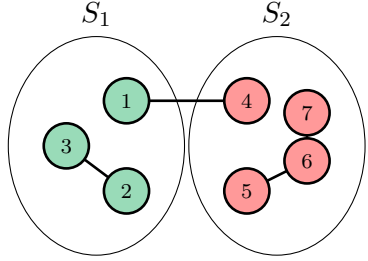
\begin{center}
		\begin{tikzpicture}[scale=0.40]
			\Vertex[size=.6,color=Green,opacity=.4,label=1]{1}
			\Vertex[x=0,y=-3.0,size=.6,color=Green,opacity=.4,label=2]{2}
			\Vertex[x=-2,y=-1.5,size=.6,color=Green,opacity=.4,label=3]{3}
			\Vertex[x=4,y=0.0,size=.6,color=red,opacity=.4,label=4]{4}
			\Vertex[x=4,y=-3.0,size=.6,color=red,opacity=.4,label=5]{5}
			\Vertex[x=6,y=-2,size=.6,color=red,opacity=.4,label=6]{6}
			\Vertex[x=6,y=-0.4,size=.6,color=red,opacity=.4,label=7]{7}
			
			\Edge[lw=1,color=black,bend=0](3)(2)
			\Edge[lw=1,color=black,bend=0](1)(4)
			\Edge[lw=1,color=black,bend=0](7)(6)
			\Edge[lw=1,color=black,bend=0](5)(6)
			
			\node[draw, ellipse, fit=(1) (2) (3), label=above:$S_1$] {};
			\node[draw, ellipse, fit=(4) (5) (6), label=above:$S_2$] {};
		\end{tikzpicture} 
	\end{center}
	Observe that $E_i^* = E_j^* = \{\{i,j\}\}$, where $E_i^*$ and $E_j^*$ are defined as in Axiom~\ref{Ax-ld} (IUN). 
	Since $\{\{i,j\}\}$ is a local unanimity graph, Axioms~\ref{Ax-ld} (IUN) and~\ref{Ax-Ba} (Ba) imply:
	\[
	\phi_i(E,\Pi) = \phi_i(E_i^*,\Pi) = \phi_j(E_j^*,\Pi) = \phi_j(E,\Pi).
	\]
	We define $\alpha := \phi_i(E,\Pi)$. Observe that the constant $\alpha$  does not depend neither on $i,j$ nor on $(E,\Pi)$ because of Axiom \ref{Ax-a} (A). Now consider a graph $(E,\Pi)$ such that $\epsilon^{\mathrm{ext}} = d_i^{\text{ext}} = p > 1$, with $d_i^{\text{int}} = 0$ for some $i \in S_1$, and $d_j^{\text{int}} = 0$ for all $j \in N(i, S_2)$ (see for example the figure below).
	
	
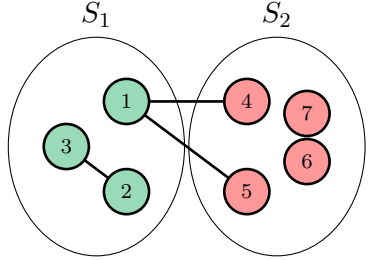
\begin{center}
		\begin{tikzpicture}[scale=0.40]
			\Vertex[size=.6,color=Green,opacity=.4,label=1]{1}
			\Vertex[x=0,y=-3.0,size=.6,color=Green,opacity=.4,label=2]{2}
			\Vertex[x=-2,y=-1.5,size=.6,color=Green,opacity=.4,label=3]{3}
			\Vertex[x=4,y=0.0,size=.6,color=red,opacity=.4,label=4]{4}
			\Vertex[x=4,y=-3.0,size=.6,color=red,opacity=.4,label=5]{5}
			\Vertex[x=6,y=-2,size=.6,color=red,opacity=.4,label=6]{6}
			\Vertex[x=6,y=-0.4,size=.6,color=red,opacity=.4,label=7]{7}
			
			\Edge[lw=1,color=black,bend=0](1)(4)
			\Edge[lw=1,color=black,bend=0](1)(5)
			\Edge[lw=1,color=black,bend=0](2)(3)
			
			\node[draw, ellipse, fit=(1) (2) (3), label=above:$S_1$] {};
			\node[draw, ellipse, fit=(4) (5) (6), label=above:$S_2$] {};
		\end{tikzpicture}
	\end{center}
	Observe that $E_i^* = E_j^*$ for all $j \in N(i,S_2)$, and $(E_i^*,\Pi)$ is a local unanimity graph.
	
	Let $\{E_h^i\}_{h=1}^p$ be the family of all subsets of edges obtained by removing from $E_i^*$ exactly $p-1$ edges connecting $i$ with nodes in $N(i, S_2)$, as allowed by Axiom~\ref{Ax-l} (Li). From the previous calculation, each $(E_h^i,\Pi)$ is such that $\phi_i(E_h^i,\Pi) = \alpha$.
	
	By Axioms~\ref{Ax-ld} (IUN) and~\ref{Ax-l} (Li), we conclude that:
	\[
	\phi_i(E,\Pi) = p \alpha.
	\]
	Moreover, by Axioms~\ref{Ax-ld} (IUN) and~\ref{Ax-Ba} (Ba),
	\[
	\sum_{j \in N(i,S_2)} \phi_j(E,\Pi) = \sum_{j \in N(i,S_2)} \phi_j(E_i^*,\Pi) = p \alpha.
	\]
	Finally, Lemma~\ref{lemma:IUN+A} implies that $\phi_j(E,\Pi) = \alpha$, for all $j \in N(i,S_2)$.
	
	\medskip

	\paragraph{STEP 2} Suppose \( r = 2 \). Take a pair \( (E, \Pi) \) such that \( \epsilon^{\mathrm{ext}} = d_i^{\text{ext}} = p \geq 1 \), \( d_i^{\text{int}} = 1 \) with \( i \in S_1 \), and \( d_j^{\text{int}} = 0 \) for all \( j \in N(i, S_2) \) (see for example the figure below).  
	
	
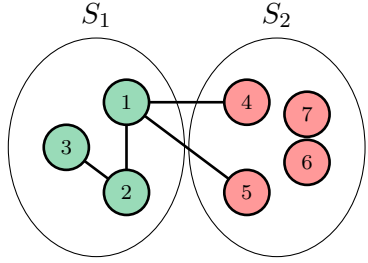
\begin{center}
		\begin{tikzpicture}[scale=0.40]
			\Vertex[size=.6,color=Green,opacity=.4,label=1]{1}
			\Vertex[x=0,y=-3.0,size=.6,color=Green,opacity=.4,label=2]{2}
			\Vertex[x=-2,y=-1.5,size=.6,color=Green,opacity=.4,label=3]{3}
			\Vertex[x=4,y=0.0,size=.6,color=red,opacity=.4,label=4]{4}
			\Vertex[x=4,y=-3.0,size=.6,color=red,opacity=.4,label=5]{5}
			\Vertex[x=6,y=-2,size=.6,color=red,opacity=.4,label=6]{6}
			\Vertex[x=6,y=-0.4,size=.6,color=red,opacity=.4,label=7]{7}
			
			\Edge[lw=1,color=black,bend=0](3)(2)
			\Edge[lw=1,color=black,bend=0](1)(4)
			\Edge[lw=1,color=black,bend=0](1)(2)
			\Edge[lw=1,color=black,bend=0](1)(5)
			\node[draw, ellipse, fit=(1) (2) (3), label=above:$S_1$] {};
			\node[draw, ellipse, fit=(4) (5) (6), label=above:$S_2$] {};
		\end{tikzpicture} 
	\end{center}

	First, note that \( (E_i^*, \Pi) \) is a local unanimity graph, so we can apply the computation from STEP 1. Then, by Axiom~\ref{Ax-ld} (IUN), we get:
	\[
	\phi_i(E, \Pi) = p \alpha.
	\]
	
	Observe that for all \( j, j' \in N(i, S_2) \), we have \( E_j^* = E_{j'}^* =: E^* \), and \( (E_j^*, \Pi) \) is a unanimity graph.  
	
	Now, take \( k \in N(i, S_1) \). By Axiom~\ref{Ax-ci} (ConI) and Lemma~\ref{lemma:ConI+A+UI}, there exists \( \beta \in \mathbb{R} \) such that:
	\begin{align*}
		\sum_{j \in N(i, S_2|E^*, \Pi)} \phi_j(E^*, \Pi)
		&= \sum_{j \in N(i, S_2|E^*, \Pi)} \phi_j(E^* \setminus \{i, k\}, \Pi) + \beta \\
		&= p \alpha + \beta.
	\end{align*}
	
	Finally, by Lemma~\ref{lemma:IUN+A}, the values \( \phi_j(E^*, \Pi) \) are all equal across \( j \in N(i, S_2) \), and since there are \( d_i(S_2) \) of them, it follows that:
	\[
	\phi_j(E, \Pi) = \alpha + \frac{\beta}{d_i(S_2)}.
	\]
	Now, we argue by induction on the internal degree of node $i$. Consider the pair $(E, \Pi)$ such that $\epsilon^{\mathrm{ext}} = d_i^{\text{ext}} = p > 1$, $d_i^{\text{int}} = m + 1$, with $i \in S_1$, and $d_j^{\text{int}} = 0$ for all $j \in N(i, S_2)$ (see for example the figure below).

	
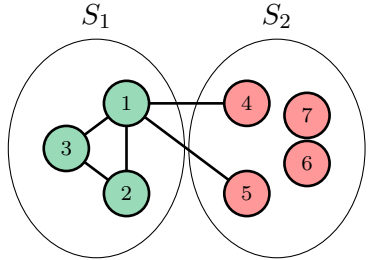
\begin{center}
		\begin{tikzpicture}[scale=0.40]
			\Vertex[size=.6,color=Green,opacity=.4,label=1]{1}
			\Vertex[x=0,y=-3.0,size=.6,color=Green,opacity=.4,label=2]{2}
			\Vertex[x=-2,y=-1.5,size=.6,color=Green,opacity=.4,label=3]{3}
			\Vertex[x=4,y=0.0,size=.6,color=red,opacity=.4,label=4]{4}
			\Vertex[x=4,y=-3.0,size=.6,color=red,opacity=.4,label=5]{5}
			\Vertex[x=6,y=-2,size=.6,color=red,opacity=.4,label=6]{6}
			\Vertex[x=6,y=-0.4,size=.6,color=red,opacity=.4,label=7]{7}
			
			\Edge[lw=1,color=black,bend=0](3)(2)
			\Edge[lw=1,color=black,bend=0](1)(4)
			\Edge[lw=1,color=black,bend=0](1)(2)
			\Edge[lw=1,color=black,bend=0](1)(3)
			\Edge[lw=1,color=black,bend=0](1)(5)
			\node[draw, ellipse, fit=(1) (2) (3), label=above:$S_1$] {};
			\node[draw, ellipse, fit=(4) (5) (6), label=above:$S_2$] {};
		\end{tikzpicture} 
	\end{center}
	As before, note that $(E_i^*, \Pi)$ is a local unanimity graph, so we can apply the calculation from STEP 1. Then, by Axiom \ref{Ax-ld} (IUN), we have $\phi_i(E, \Pi) = p\alpha$. 
	
	Next, take $k \in N(i, S_1)$ and assume that $\phi_j(E \setminus \{i, k\}, \Pi) = \alpha + \frac{\beta m}{d_i(S_2)}$ for all $j \in N(i, S_2)$. 
	
	Observe that $E_j^* = E_{j'}^* =: E^*$ for all $j, j' \in N(i, S_2)$, and that $(E_j^*, \Pi)$ is a unanimity graph. By Axiom \ref{Ax-ci} (ConI) and Lemma \ref{lemma:ConI+A+UI}, we now have the following:
	
	\begin{align*}
		\sum_{j \in N(i,S_2|E^*,\Pi)}\phi_j(E^*,\Pi)&= \sum_{j \in N(i,S_2|E^*,\Pi)}\phi_{j}(E^*\setminus \{i,k\},\Pi)+\beta\\
		&=  p \alpha + \frac{\beta m p}{d_i(S_2)} + \beta = p \alpha + \beta(m+1).
	\end{align*}
	Then by Lemma \ref{lemma:IUN+A} we know that 
	$$ \phi_j(E,\Pi)=   \alpha + \frac{\beta(m+1)}{d_i(S_2)}=   \alpha + \beta\frac{d_i(S_1)}{d_i (S_2)}=\varphi^{(\alpha,\beta)}_j(E,\Pi).$$

	\paragraph{STEP 3} 
	Suppose $r=2$. Take an arbitrary graph with a priori unions $(E,\Pi)$ such that $\epsilon^{\mathrm{ext}}=d_i^{\text{ext}}=p\geq1$,  
	with $i\in S_1$  
	(see for example the figure below).
	
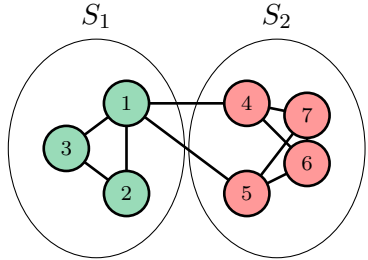
\begin{center}
		\begin{tikzpicture}[scale=0.40]
			\Vertex[size=.6,color=Green,opacity=.4,label=1]{1}
			\Vertex[x=0,y=-3.0,size=.6,color=Green,opacity=.4,label=2]{2}
			\Vertex[x=-2,y=-1.5,size=.6,color=Green,opacity=.4,label=3]{3}
			\Vertex[x=4,y=0.0,size=.6,color=red,opacity=.4,label=4]{4}
			\Vertex[x=4,y=-3.0,size=.6,color=red,opacity=.4,label=5]{5}
			\Vertex[x=6,y=-2,size=.6,color=red,opacity=.4,label=6]{6}
			\Vertex[x=6,y=-0.4,size=.6,color=red,opacity=.4,label=7]{7}
			
			\Edge[lw=1,color=black,bend=0](3)(2)
			\Edge[lw=1,color=black,bend=0](1)(4)
			\Edge[lw=1,color=black,bend=0](1)(2)
			\Edge[lw=1,color=black,bend=0](1)(3)
			\Edge[lw=1,color=black,bend=0](1)(5)
			\Edge[lw=1,color=black,bend=0](4)(7)
			\Edge[lw=1,color=black,bend=0](6)(4)
			\Edge[lw=1,color=black,bend=0](7)(5)
			\Edge[lw=1,color=black,bend=0](7)(6)
			\Edge[lw=1,color=black,bend=0](5)(6)
			\node[draw, ellipse, fit=(1) (2) (3), label=above:$S_1$] {};
			\node[draw, ellipse, fit=(4) (5) (6), label=above:$S_2$] {};
		\end{tikzpicture} 
	\end{center}
	
	First, observe that $E_j^* = E_{j'}^* =: E^*$ for all $j, j' \in N(i,S_2)$, and that $(E_j^*, \Pi)$ is a unanimity graph. Therefore, we can apply the calculations from STEP 2 to obtain $\phi_j(E, \Pi) = \varphi^{(\alpha,\beta)}_j(E, \Pi)$. 
	
	Next, consider the family $\{ E_h^i \}_{h=1}^p$ of all possible edge sets obtained by removing the $p-1$ edges connecting node $i$ to the nodes in $N(i, S_2|E_i^*, \Pi)$. Using Axiom \ref{Ax-l} (Li), we define $j_h$ as the unique node in $S_2$ that is connected to $i$ in $E_h^i$. The graphs $(E_h^i, \Pi)$ are unanimity graphs, so we can apply STEP 2, where node $i$ plays the role of the nodes $j$. 
	
	Thus, we have:
	\begin{align*}
		\phi_i(E, \Pi) &= \sum_{h=1}^p \phi_i(E_h^i, \Pi) = \sum_{h=1}^p \alpha + \beta \frac{d_{j_h}(S_2)}{d_{j_h}(S_1)} \\
		&= \alpha d_i(S_2) + \beta \sum_{j \in N(i,S_2)} \frac{d_j(S_2)}{d_j(S_1)} = \varphi^{(\alpha,\beta)}_i(E, \Pi).
	\end{align*}

	\paragraph{STEP 4} Suppose $r=2$.
	Consider an arbitrary graph with a priori unions $(E,\Pi)$ such that $\epsilon^{\mathrm{ext}}>1$ (see for example the figure below). 
	
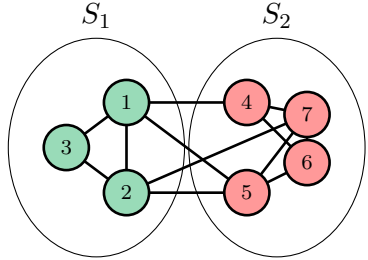
\begin{center}
		\begin{tikzpicture}[scale=0.40]
			\Vertex[size=.6,color=Green,opacity=.4,label=1]{1}
			\Vertex[x=0,y=-3.0,size=.6,color=Green,opacity=.4,label=2]{2}
			\Vertex[x=-2,y=-1.5,size=.6,color=Green,opacity=.4,label=3]{3}
			\Vertex[x=4,y=0.0,size=.6,color=red,opacity=.4,label=4]{4}
			\Vertex[x=4,y=-3.0,size=.6,color=red,opacity=.4,label=5]{5}
			\Vertex[x=6,y=-2,size=.6,color=red,opacity=.4,label=6]{6}
			\Vertex[x=6,y=-0.4,size=.6,color=red,opacity=.4,label=7]{7}
			
			\Edge[lw=1,color=black,bend=0](3)(2)
			\Edge[lw=1,color=black,bend=0](1)(4)
			\Edge[lw=1,color=black,bend=0](1)(2)
			\Edge[lw=1,color=black,bend=0](1)(3)
			\Edge[lw=1,color=black,bend=0](1)(5)
			\Edge[lw=1,color=black,bend=0](2)(7)
			\Edge[lw=1,color=black,bend=0](2)(5)
			\Edge[lw=1,color=black,bend=0](4)(7)
			\Edge[lw=1,color=black,bend=0](6)(4)
			\Edge[lw=1,color=black,bend=0](7)(5)
			\Edge[lw=1,color=black,bend=0](7)(6)
			\Edge[lw=1,color=black,bend=0](5)(6)
			\node[draw, ellipse, fit=(1) (2) (3), label=above:$S_1$] {};
			\node[draw, ellipse, fit=(4) (5) (6), label=above:$S_2$] {};
		\end{tikzpicture} 
	\end{center}
	
	For each $i\in N$ such that $d_i^{\text{ext}}=p>1$, applying Axiom \ref{Ax-l} (Li) it holds 
	$$
	\phi_i(E,\Pi) = \sum_{h=1}^p \phi_i(E_h^i,\Pi). 
	$$
	
	Using Axiom \ref{Ax-ld} (IUN) again on node $i$ it holds
	$$
	\phi_i(E_h^i,\Pi) = \phi_i((E_h^i)^*_i,\Pi). 
	$$
	Observe that the graph $((E_h^i)^*_i,\Pi)$ satisfies the hypothesis of STEP 2, then $\phi_i((E_h^i)^*_i,\Pi)=\alpha+\beta\frac{d_{j_h}(S^{-i})}{d_{j_h}(S^i)}$, where $j_h$ is the unique node in $S^{-i}$ connect to $i$ in $E_{h}^i$. We have
	$$
	\phi_i(E,\Pi) = \alpha d_i(S^{-i}) + \beta\sum_{j\in N(i,S^{-i})} \frac{d_j(S^{-i})}{d_j(S^i)} = \varphi_i^{(\alpha,\beta)}(E,\Pi). 
	$$
	
	\paragraph{STEP 5} Take $r>2$. Consider a generic graph with a priori unions $(E,\Pi)\in\mathcal{G}(N)$ with $\Pi\in\mathbf{\Pi}_{r}(N)$ and $r\leq n$ (see for example the figure below). 
	
	
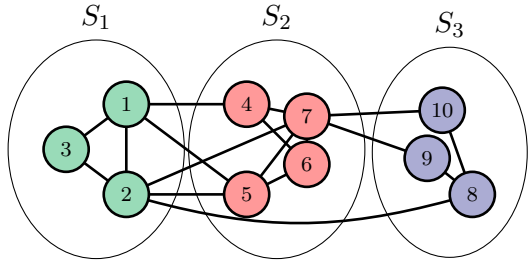
\begin{center}
		\begin{tikzpicture}[scale=0.40]
			\Vertex[size=.6,color=Green,opacity=.4,label=1]{1}
			\Vertex[x=0,y=-3.0,size=.6,color=Green,opacity=.4,label=2]{2}
			\Vertex[x=-2,y=-1.5,size=.6,color=Green,opacity=.4,label=3]{3}
			\Vertex[x=4,y=0.0,size=.6,color=red,opacity=.4,label=4]{4}
			\Vertex[x=4,y=-3.0,size=.6,color=red,opacity=.4,label=5]{5}
			\Vertex[x=6,y=-2,size=.6,color=red,opacity=.4,label=6]{6}
			\Vertex[x=6,y=-0.4,size=.6,color=red,opacity=.4,label=7]{7}
			\Vertex[x=11.5,y=-3.0,size=.6,color=Blue,opacity=.4,label=8]{8}
			\Vertex[x=10,y=-1.8,size=.6,color=Blue,opacity=.4,label=9]{9}
			\Vertex[x=10.5,y=-0.2,size=.6,color=Blue,opacity=.4,label=10]{10}
			\Edge[lw=1,color=black,bend=0](3)(2)
			\Edge[lw=1,color=black,bend=0](1)(4)
			\Edge[lw=1,color=black,bend=0](1)(2)
			\Edge[lw=1,color=black,bend=0](1)(3)
			\Edge[lw=1,color=black,bend=0](1)(5)
			\Edge[lw=1,color=black,bend=0](2)(7)
			\Edge[lw=1,color=black,bend=0](2)(5)
			\Edge[lw=1,color=black,bend=0](4)(7)
			\Edge[lw=1,color=black,bend=0](6)(4)
			\Edge[lw=1,color=black,bend=0](7)(5)
			\Edge[lw=1,color=black,bend=0](7)(6)
			\Edge[lw=1,color=black,bend=0](5)(6)
			\Edge[lw=1,color=black,bend=0](10)(7)
			\Edge[lw=1,color=black,bend=0](9)(7)
			\Edge[lw=1,color=black,bend=-15](2)(8)
			\Edge[lw=1,color=black,bend=0](9)(8)
			\Edge[lw=1,color=black,bend=0](10)(8)
			\node[draw, ellipse, fit=(1) (2) (3), label=above:$S_1$] {};
			\node[draw, ellipse, fit=(4) (5) (6), label=above:$S_2$] {};
			\node[draw, ellipse, fit=(8) (9) (10), label=above:$S_3$] {};
		\end{tikzpicture} 
	\end{center}
	
	Without loss of generality, assume that \( i \in S_1 \) and \( d_i^{\text{ext}} = p > 0 \). By Axiom \ref{Ax-l} (Li), we have
	
	\[
	\phi_i(E,\Pi) = \sum_{h=1}^p \phi_i(E_h^i,\Pi).
	\]
	
	Next, applying Axiom \ref{Ax-ld} (IUN) again to node \( i \), we obtain
	
	\[
	\phi_i(E_h,\Pi) = \phi_i((E_h^i)_i^*, \Pi).
	\]
	
	Let \( j_h \) be the unique node in \( S^{-i} \) that is connected to \( i \) in \( E_h^i \). Notice that the uions, other than \( S^{j_h} \), do not have links to \( S_1 \). Therefore, by applying Axiom \ref{Ax-UI} (UI), we get
	
	\[
	\phi_i((E_h^i)_i^*, \Pi) = \phi_i((E_h^i)_i^*, \{S_1, \cup_{j \neq 1} S^j\}) = \alpha + \beta \sum_{j \in N(i,S^{-i})} \frac{d_j(S^{-i})}{d_j(S_1)}.
	\]
	
	Thus, we have
	
	\[
	\phi_i(E, \Pi) = \alpha d_i(S^{-i}) + \beta \sum_{j \in N(i, S^{-i})} \frac{d_j(S^{-i})}{d_j(S_1)} = \varphi_i^{(\alpha,\beta)}(E, \Pi).
	\]
	
	The proof is then complete.
\end{proof}

\begin{proof}[\bf Proof of Theorem \ref{mainp}]
	The proof retraces the same steps of the proof of Theorem \ref{main}. Observe that at STEP 1 we get $\phi_i(E,\Pi)=\alpha\geq 0$, by Axiom \ref{Ax-pen} (EP). Moreover, if in Lemma \ref{lemma:ConI+A+UI} we substitute Axiom \ref{Ax-ci} (ConI) with Axiom \ref{Ax-pci} (PConI), then we obtain $\beta\geq0$.
\end{proof}

\bibliographystyle{elsarticle-harv}
\bibliography{References}

\begin{thebibliography}{30}
\expandafter\ifx\csname natexlab\endcsname\relax\def\natexlab#1{#1}\fi
\providecommand{\url}[1]{\texttt{#1}}
\providecommand{\href}[2]{#2}
\providecommand{\path}[1]{#1}
\providecommand{\DOIprefix}{doi:}
\providecommand{\ArXivprefix}{arXiv:}
\providecommand{\URLprefix}{URL: }
\providecommand{\Pubmedprefix}{pmid:}
\providecommand{\doi}[1]{\href{http://dx.doi.org/#1}{\path{#1}}}
\providecommand{\Pubmed}[1]{\href{pmid:#1}{\path{#1}}}
\providecommand{\bibinfo}[2]{#2}
\ifx\xfnm\relax \def\xfnm[#1]{\unskip,\space#1}\fi
\bibitem[{Alonso-Meijide et~al.(2009)Alonso-Meijide, {\'A}lvarez-Mozos and
  Fiestras-Janeiro}]{alonso2009values}
\bibinfo{author}{Alonso-Meijide, J.M.}, \bibinfo{author}{{\'A}lvarez-Mozos,
  M.}, \bibinfo{author}{Fiestras-Janeiro, M.G.}, \bibinfo{year}{2009}.
\newblock \bibinfo{title}{Values of games with graph restricted communication
  and a priori unions}.
\newblock \bibinfo{journal}{Mathematical Social Sciences} \bibinfo{volume}{58},
  \bibinfo{pages}{202--213}.
\bibitem[{Arora and Ceccagnoli(2006)}]{arora2006patent}
\bibinfo{author}{Arora, A.}, \bibinfo{author}{Ceccagnoli, M.},
  \bibinfo{year}{2006}.
\newblock \bibinfo{title}{Patent protection, complementary assets, and firms'
  incentives for technology licensing}.
\newblock \bibinfo{journal}{Management science} \bibinfo{volume}{52},
  \bibinfo{pages}{293--308}.
\bibitem[{Bilbao(2012)}]{bilbao2012cooperative}
\bibinfo{author}{Bilbao, J.M.}, \bibinfo{year}{2012}.
\newblock \bibinfo{title}{Cooperative games on combinatorial structures}.
  volume~\bibinfo{volume}{26}.
\newblock \bibinfo{publisher}{Springer Science \& Business Media},
  \bibinfo{address}{New York}.
\bibitem[{Billot and Thisse(2005)}]{billot2005share}
\bibinfo{author}{Billot, A.}, \bibinfo{author}{Thisse, J.F.},
  \bibinfo{year}{2005}.
\newblock \bibinfo{title}{How to share when context matters: The {M}{\"o}bius
  value as a generalized solution for cooperative games}.
\newblock \bibinfo{journal}{Journal of Mathematical Economics}
  \bibinfo{volume}{41}, \bibinfo{pages}{1007--1029}.
\bibitem[{Brandenburger and Nalebuff(2011)}]{brandenburger2011co}
\bibinfo{author}{Brandenburger, A.M.}, \bibinfo{author}{Nalebuff, B.J.},
  \bibinfo{year}{2011}.
\newblock \bibinfo{title}{Co-opetition}.
\newblock \bibinfo{publisher}{Crown Currency}.
\bibitem[{Van~den Brink et~al.(2015)Van~den Brink, Van~der Laan and
  Moes}]{van2015values}
\bibinfo{author}{Van~den Brink, R.}, \bibinfo{author}{Van~der Laan, G.},
  \bibinfo{author}{Moes, N.}, \bibinfo{year}{2015}.
\newblock \bibinfo{title}{Values for transferable utility games with coalition
  and graph structure}.
\newblock \bibinfo{journal}{Top} \bibinfo{volume}{23}, \bibinfo{pages}{77--99}.
\bibitem[{Van~den Brink and Rusinowska(2022)}]{van2022degree}
\bibinfo{author}{Van~den Brink, R.}, \bibinfo{author}{Rusinowska, A.},
  \bibinfo{year}{2022}.
\newblock \bibinfo{title}{The degree measure as utility function over positions
  in graphs and digraphs}.
\newblock \bibinfo{journal}{European Journal of Operational Research}
  \bibinfo{volume}{299}, \bibinfo{pages}{1033--1044}.
\bibitem[{Fagerberg et~al.(2005)Fagerberg, Mowery and
  Nelson}]{fagerberg2005oxford}
\bibinfo{author}{Fagerberg, J.}, \bibinfo{author}{Mowery, D.C.},
  \bibinfo{author}{Nelson, R.R.}, \bibinfo{year}{2005}.
\newblock \bibinfo{title}{The Oxford handbook of innovation}.
\newblock \bibinfo{publisher}{Oxford University Press}.
\bibitem[{Fosfuri(2006)}]{fosfuri2006licensing}
\bibinfo{author}{Fosfuri, A.}, \bibinfo{year}{2006}.
\newblock \bibinfo{title}{The licensing dilemma: understanding the determinants
  of the rate of technology licensing}.
\newblock \bibinfo{journal}{Strategic Management Journal} \bibinfo{volume}{27},
  \bibinfo{pages}{1141--1158}.
\bibitem[{Gilsing et~al.(2008)Gilsing, Nooteboom, Vanhaverbeke, Duysters and
  Van Den~Oord}]{gilsing2008network}
\bibinfo{author}{Gilsing, V.}, \bibinfo{author}{Nooteboom, B.},
  \bibinfo{author}{Vanhaverbeke, W.}, \bibinfo{author}{Duysters, G.},
  \bibinfo{author}{Van Den~Oord, A.}, \bibinfo{year}{2008}.
\newblock \bibinfo{title}{Network embeddedness and the exploration of novel
  technologies: Technological distance, betweenness centrality and density}.
\newblock \bibinfo{journal}{Research policy} \bibinfo{volume}{37},
  \bibinfo{pages}{1717--1731}.
\bibitem[{Grabisch(1997)}]{grabisch1997k}
\bibinfo{author}{Grabisch, M.}, \bibinfo{year}{1997}.
\newblock \bibinfo{title}{K-order additive discrete fuzzy measures and their
  representation}.
\newblock \bibinfo{journal}{Fuzzy Sets and Systems} \bibinfo{volume}{92},
  \bibinfo{pages}{167--189}.
\bibitem[{Grabisch(2006)}]{grabisch2006capacities}
\bibinfo{author}{Grabisch, M.}, \bibinfo{year}{2006}.
\newblock \bibinfo{title}{Capacities and games on lattices: A survey of
  results}.
\newblock \bibinfo{journal}{International Journal of Uncertainty, Fuzziness and
  Knowledge-Based Systems} \bibinfo{volume}{14}, \bibinfo{pages}{371--392}.
\bibitem[{Harsanyi(1959)}]{harsanyi1959bargaining}
\bibinfo{author}{Harsanyi, J.C.}, \bibinfo{year}{1959}.
\newblock \bibinfo{title}{A bargaining model for the cooperative n-person
  game}, in: \bibinfo{editor}{Tucker, A.W.}, \bibinfo{editor}{Luce, R.D.}
  (Eds.), \bibinfo{booktitle}{Contributions to the Theory of Games (AM-40)}.
  \bibinfo{publisher}{Princeton University Press}, \bibinfo{address}{Princeton,
  New Jersey}. volume~\bibinfo{volume}{IV}, pp. \bibinfo{pages}{325--356}.
\bibitem[{Hart and Kurz(1983)}]{hart1983endogenous}
\bibinfo{author}{Hart, S.}, \bibinfo{author}{Kurz, M.}, \bibinfo{year}{1983}.
\newblock \bibinfo{title}{Endogenous formation of coalitions}.
\newblock \bibinfo{journal}{Econometrica: Journal of the Econometric Society} ,
  \bibinfo{pages}{1047--1064}.
\bibitem[{Hart and Mas-Colell(1988)}]{hart1988potential}
\bibinfo{author}{Hart, S.}, \bibinfo{author}{Mas-Colell, A.},
  \bibinfo{year}{1988}.
\newblock \bibinfo{title}{The potential of the {S}hapley value}, in:
  \bibinfo{editor}{Roth, A.E.} (Ed.), \bibinfo{booktitle}{The {S}hapley value:
  Essays in honor of Lloyd S. {S}hapley}. \bibinfo{publisher}{Cambridge
  University Press}, \bibinfo{address}{Cambridge (UK)}, pp.
  \bibinfo{pages}{127--137}.
\bibitem[{Hellman and Peretz(2018)}]{hellman2018values}
\bibinfo{author}{Hellman, Z.}, \bibinfo{author}{Peretz, R.},
  \bibinfo{year}{2018}.
\newblock \bibinfo{title}{Values for cooperative games over graphs and games
  with inadmissible coalitions}.
\newblock \bibinfo{journal}{Games and Economic Behavior} \bibinfo{volume}{108},
  \bibinfo{pages}{22--36}.
\bibitem[{Karos and Peters(2015)}]{karos2015indirect}
\bibinfo{author}{Karos, D.}, \bibinfo{author}{Peters, H.},
  \bibinfo{year}{2015}.
\newblock \bibinfo{title}{Indirect control and power in mutual control
  structures}.
\newblock \bibinfo{journal}{Games and Economic Behavior} \bibinfo{volume}{92},
  \bibinfo{pages}{150--165}.
\bibitem[{Khmelnitskaya et~al.(2016)Khmelnitskaya, Sel{\c c}uk and
  Talman}]{khmelnitskaya2016shapley}
\bibinfo{author}{Khmelnitskaya, A.}, \bibinfo{author}{Sel{\c c}uk, {\"O}.},
  \bibinfo{author}{Talman, D.}, \bibinfo{year}{2016}.
\newblock \bibinfo{title}{The {S}hapley value for directed graph games}.
\newblock \bibinfo{journal}{Operations Research Letters} \bibinfo{volume}{44},
  \bibinfo{pages}{143--147}.
\bibitem[{Li and Shan(2020a)}]{li2020efficient}
\bibinfo{author}{Li, D.L.}, \bibinfo{author}{Shan, E.}, \bibinfo{year}{2020}a.
\newblock \bibinfo{title}{Efficient quotient extensions of the {M}yerson
  value}.
\newblock \bibinfo{journal}{Annals of Operations Research}
  \bibinfo{volume}{292}, \bibinfo{pages}{171--181}.
\bibitem[{Li and Shan(2020b)}]{li2020myerson}
\bibinfo{author}{Li, D.L.}, \bibinfo{author}{Shan, E.}, \bibinfo{year}{2020}b.
\newblock \bibinfo{title}{The {M}yerson value for directed graph games}.
\newblock \bibinfo{journal}{Operations Research Letters} \bibinfo{volume}{48},
  \bibinfo{pages}{142--146}.
\bibitem[{Li and Shan(2019)}]{li2019myerson}
\bibinfo{author}{Li, D.L.}, \bibinfo{author}{Shan, E.F.}, \bibinfo{year}{2019}.
\newblock \bibinfo{title}{The {M}yerson value on local structures of
  coalitions}.
\newblock \bibinfo{journal}{Journal of the Operations Research Society of
  China} \bibinfo{volume}{7}, \bibinfo{pages}{461--473}.
\bibitem[{Myerson(1977)}]{myerson1977graphs}
\bibinfo{author}{Myerson, R.}, \bibinfo{year}{1977}.
\newblock \bibinfo{title}{Graphs and cooperation in games}.
\newblock \bibinfo{journal}{Mathematics of Operations Research} ,
  \bibinfo{pages}{225--229}.
\bibitem[{Owen(1977)}]{owen1977values}
\bibinfo{author}{Owen, G.}, \bibinfo{year}{1977}.
\newblock \bibinfo{title}{Values of games with a priori unions}, in:
  \bibinfo{editor}{Henn, R.}, \bibinfo{editor}{Moeschlin, O.} (Eds.),
  \bibinfo{booktitle}{Mathematical Economics and Game Theory},
  \bibinfo{publisher}{Springer Berlin Heidelberg}, \bibinfo{address}{Berlin,
  Heidelberg}. pp. \bibinfo{pages}{76--88}.
\bibitem[{Owen(1981)}]{owen1981modification}
\bibinfo{author}{Owen, G.}, \bibinfo{year}{1981}.
\newblock \bibinfo{title}{Modification of the {B}anzhaf-{C}oleman index for
  games with a priori unions}, in: \bibinfo{booktitle}{Power, voting, and
  voting power}. \bibinfo{publisher}{Springer}, \bibinfo{address}{Heidelberg},
  pp. \bibinfo{pages}{232--238}.
\bibitem[{Shapiro(1985)}]{Shapiro1985patent}
\bibinfo{author}{Shapiro, C.}, \bibinfo{year}{1985}.
\newblock \bibinfo{title}{Patent licensing and r \& d rivalry}.
\newblock \bibinfo{journal}{The American economic review} \bibinfo{volume}{75},
  \bibinfo{pages}{25--30}.
\bibitem[{Shapley(1953)}]{shapley1953b}
\bibinfo{author}{Shapley, L.S.}, \bibinfo{year}{1953}.
\newblock \bibinfo{title}{Additive and Non-Additive Set Functions}.
\newblock Ph.D. thesis. Princeton University.
\newblock \bibinfo{note}{Ph.D. thesis}.
\bibitem[{Szczepa{\'n}ski et~al.(2014)Szczepa{\'n}ski, Michalak and
  Wooldridge}]{szczepanski2014centrality}
\bibinfo{author}{Szczepa{\'n}ski, P.L.}, \bibinfo{author}{Michalak, T.P.},
  \bibinfo{author}{Wooldridge, M.}, \bibinfo{year}{2014}.
\newblock \bibinfo{title}{A centrality measure for networks with community
  structure based on a generalization of the {O}wen value}, in:
  \bibinfo{booktitle}{ECAI 2014}, \bibinfo{publisher}{IOS Press},
  \bibinfo{address}{Amsterdam}. pp. \bibinfo{pages}{867--872}.
\bibitem[{Van Den~Brink et~al.(2008)Van Den~Brink, Borm, Hendrickx and
  Owen}]{van2008characterizations}
\bibinfo{author}{Van Den~Brink, R.}, \bibinfo{author}{Borm, P.},
  \bibinfo{author}{Hendrickx, R.}, \bibinfo{author}{Owen, G.},
  \bibinfo{year}{2008}.
\newblock \bibinfo{title}{Characterizations of the $\beta$-and the degree
  network power measure}.
\newblock \bibinfo{journal}{Theory and Decision} \bibinfo{volume}{64},
  \bibinfo{pages}{519--536}.
\bibitem[{Van Den~Brink and Rusinowska(2024)}]{van2024degree}
\bibinfo{author}{Van Den~Brink, R.}, \bibinfo{author}{Rusinowska, A.},
  \bibinfo{year}{2024}.
\newblock \bibinfo{title}{Degree centrality, von neumann--morgenstern expected
  utility and externalities in networks}.
\newblock \bibinfo{journal}{European Journal of Operational Research}
  \bibinfo{volume}{319}, \bibinfo{pages}{669--677}.
\bibitem[{V{\'a}zquez-Brage et~al.(1996)V{\'a}zquez-Brage, Garc{\i}a-Jurado and
  Carreras}]{vazquez1996owen}
\bibinfo{author}{V{\'a}zquez-Brage, M.}, \bibinfo{author}{Garc{\i}a-Jurado,
  I.}, \bibinfo{author}{Carreras, F.}, \bibinfo{year}{1996}.
\newblock \bibinfo{title}{The owen value applied to games with graph-restricted
  communication}.
\newblock \bibinfo{journal}{Games and Economic Behavior} \bibinfo{volume}{12},
  \bibinfo{pages}{42--53}.

\end{thebibliography}

\end{document}